\documentclass{article}
\usepackage[a4paper, total={6in, 9in}]{geometry}
\usepackage[utf8]{inputenc}
\usepackage{amsmath,amsfonts}
\usepackage{booktabs} 
\usepackage{caption} 
\usepackage{subcaption} 
\usepackage{graphicx}
\usepackage{pgfplots}
\usepackage[all]{nowidow}
\usepackage[utf8]{inputenc}
\usepackage{tikz}
\usetikzlibrary{er,positioning,bayesnet}
\usepackage{multicol}
\usepackage{algpseudocode,algorithm,algorithmicx}
\usepackage{physics}
\usepackage{amssymb}
\usepackage{amsthm}
\usepackage{hyperref}

\newtheorem{theorem}{Proposition}[subsection]
\newtheorem{remark}{Remark}[subsection]
\definecolor{blue}{HTML}{1F77B4}
\definecolor{orange}{HTML}{FF7F0E}
\definecolor{green}{HTML}{2CA02C}

\pgfplotsset{compat=1.14}

\newcommand{\no}{\noindent}
\begin{document}
\title{A High-Order Asymptotic Analysis of the Benjamin-Feir Instability Spectrum in Arbitrary Depth}
%
%
\author{\large Ryan P. Creedon$^1$ and Bernard Deconinck$^2$}
%
%

\date{\normalsize \today}

\maketitle              
\vspace*{-0.7cm}
{\begin{center} {\scriptsize\noindent $^1$Department of Applied Mathematics, University of Washington, Seattle, WA, USA, \href{mailto:creedon@uw.edu}{creedon@uw.edu}\\ \vspace*{-0.0cm} \noindent $^2$Department of Applied Mathematics, University of Washington, Seattle, WA, USA,  \href{mailto:deconinc@uw.edu}{deconinc@uw.edu}}\end{center}}

\begin{abstract}
We investigate the Benjamin-Feir (or modulational) instability of Stokes waves, \emph{i.e.}, small-amplitude, one-dimensional periodic gravity waves of permanent form and constant velocity, in water of finite and infinite depth.
We develop a perturbation method to describe to high-order accuracy the unstable spectral elements associated with this instability, obtained by linearizing Euler's equations about the small-amplitude Stokes waves. These unstable elements form a figure-eight curve centered at the origin of the complex spectral plane, which is parameterized by a Floquet exponent. Our asymptotic expansions of this figure-eight are in excellent agreement with numerical computations as well as recent rigorous results by Berti, Maspero, \emph{\&} Ventura \cite{bertiMasperoVentura2021,bertiMasperoVentura2022}. From our expansions, we derive high-order estimates for the growth rates of the Benjamin-Feir instability and for the parameterization of the Benjamin-Feir figure-eight curve with respect to the Floquet exponent. We are also able to compare the Benjamin-Feir and high-frequency instability spectra analytically for the first time, revealing three different regimes of the Stokes waves, depending on the predominant instability. \\\\
\noindent {\bf Keywords}: Euler's equations, \and Stokes waves, \and Benjamin-Feir (or modulational) instability, \and spectral instability, \and perturbation methods
\end{abstract}
%


%
\section{Introduction}
In 1847, Stokes \cite{stokes1847} derived a formal asymptotic expansion for the small-amplitude, periodic traveling wave solutions of the full water wave equations in infinite depth, see Figure \ref{fig1} for a schematic. Seventy-five years later, Nekrasov \cite{nekrasov1921} and Levi-Civita \cite{levicivita1925} proved the validity of these expansions if the amplitude $\varepsilon$ of the waves is sufficiently small. Not long after, Struik \cite{struik1926} extended this analysis to finite depth. The stability with respect to sideband perturbations of these solutions, now known as Stokes waves, was investigated first experimentally by Benjamin \emph{\&} Feir \cite{benjaminFeir1967} in the late 1960s and immediately after supported by independent formal calculations by Benjamin \cite{benjamin1967} and Whitham \cite{whitham1967b} using distinct methods. Both calculations suggest that Stokes waves are modulationally unstable with respect to longitudinal sideband perturbations provided $\kappa h > \alpha_{BW}$, where $\kappa$ is the wavenumber of the Stokes waves, $h$ is the constant depth of the water, and $\alpha_{BW} = 1.3627827567...$. This instability is now known as the Benjamin-Feir or modulational instability. 

In the years since the pioneering work of Benjamin and Whitham, several papers have explored the Benjamin-Feir instability experimentally, numerically, and analytically. It is impossible to summarize all these works here, rather we highlight a few works that are most relevant to our investigation. For a more comprehensive history of the Benjamin-Feir instability, see \cite{craik2004,grimshaw2005,yuenLake1980} and, in different contexts, \cite{deconinckOliveras2011,korotkevichetal2016,zakharovOstrovsky2009}. 
\begin{figure}[tb]
    \centering
    \includegraphics[height=3.2cm,width=10.6cm]{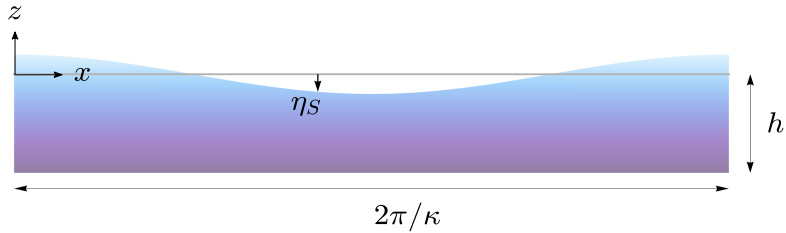} ~\\
    \caption{\small A schematic of a $2\pi/\kappa$-periodic Stokes wave $\eta_S$ in finite depth $h$. The wave travels uniformly to the left or right, depending on its velocity. In the case of infinite depth, Stokes derived an asymptotic expansion for $\eta_S$ and its velocity as power series in a small parameter $\varepsilon$ related to the amplitude of the wave.}
    \label{fig1}
\end{figure}

Beginning in the 1970s, Bryant \cite{bryant1974,bryant1978} studied the stability of Stokes waves in shallow water ($\kappa h<\alpha_{BW}$) with respect to co-periodic longitudinal and transverse perturbations, respectively. Around the same time, Longuet-Higgins \cite{longuet-higgins1978a,longuet-higgins1978b} considered the stability of Stokes waves with respect to sub- and super-harmonic longitudinal perturbations in infinitely deep water. 
In 1982, McLean \cite{mclean1982} generalized the results of Bryant and Longuet-Higgins by investigating the stability of Stokes waves in finite depth with respect to sub- and super-harmonic transverse perturbations. 

All these results are obtained by linearizing the full water wave equations about a Stokes wave and numerically solving a resulting spectral problem. The elements of the spectrum $\lambda$ give the exponential growth rates of the perturbations, and the collection of all such $\lambda$ is called the stability spectrum of the Stokes wave. The Benjamin-Feir instability occurs when four eigenvalues of the stability spectrum collide at the origin of the complex spectral plane and separate, each eigenvalue occupying a distinct quadrant, as the amplitude $\varepsilon$ of the Stokes wave increases. This is a consequence of the quadrafold symmetry of the spectrum due to the Hamiltonian structure of the water wave equations \cite{haragusKapitula2008,zakharov1968}.

Despite numerical evidence of these eigenvalues close to the origin in the complex spectral plane, a proof of their existence did not appear until 1995. In the classic work of Bridges \emph{\&} Mielke \cite{bridgesMielke1995}, techniques from spatial dynamics and center manifold theory are used to prove the existence of unstable eigenvalues close to the origin whenever $\kappa h > \alpha_{BW}$ and $\varepsilon$ is sufficiently small. The proof is valid for any finite depth but fails in water of infinite depth. Only in the past two years have Nguyen \emph{\&} Strauss \cite{nguyenStrauss2020} developed a proof based on Lyapunov-Schmidt reduction that works in both finite and infinite depth. Another proof based on periodic Evans functions appeared in the literature more recently \cite{hurYang2020}.

In 2011, Deconinck \emph{\&} Oliveras \cite{deconinckOliveras2011} revisited the numerical computations of McLean and others, also allowing for more general quasi-periodic perturbations parameterized by a Floquet exponent $\mu$ \cite{deconinckKutz2006}. For fixed $\varepsilon$, this results in a continuous stability spectrum that can be decomposed as a union over all Floquet exponents of point spectra containing a countable number of finite-multiplicity eigenvalues~\cite{kapitulaPromislow2013}. To our knowledge, \cite{deconinckOliveras2011} is the first work displaying full stability spectra of Stokes waves in finite and infinite depth (to within machine precision) that exhibits a figure-eight curve at the origin of the complex spectral plane associated with the Benjamin-Feir instability, see a schematic in Figure~\ref{fig2}. Evidence of instabilities away from the origin, referred to as high-frequency instabilities, was also demonstrated in \cite{deconinckOliveras2011}. These high-frequency instabilities have been analytically explored with Trichtchenko in \cite{creedonDeconinckTrichtchenko2022} and by Hur \emph{\&} Yang in \cite{hurYang2020}. 

\begin{figure}[tb]
    \centering
    \includegraphics[height=6.3cm,width=6.9cm]{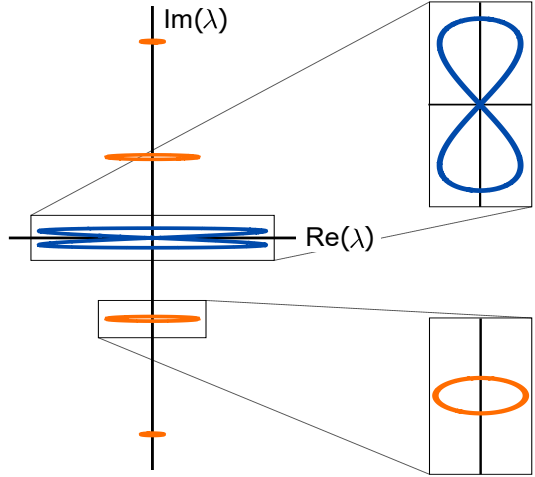}  \vspace*{4mm}
    \caption{\small A schematic of the stability spectrum for a small-amplitude Stokes wave in sufficiently deep water. The Benjamin-Feir figure-eight curve is colored blue, while the high-frequency instabilities are colored orange. In this manuscript, we are concerned with high-order asymptotic approximations of the blue figure-eight curve. For the corresponding approximations of the high-frequency instabilities in orange, see \cite{creedonDeconinckTrichtchenko2022}.}
    \label{fig2}
\end{figure}

In recent months, seminal work by Berti, Maspero, \emph{\&} Ventura \cite{bertiMasperoVentura2021,bertiMasperoVentura2022} has confirmed the existence of the Benjamin-Feir figure-eight in finite and infinite depth, provided $\kappa h > \alpha_{BW}$ and $\varepsilon$ is sufficiently small. The proof of both cases relies on the Hamiltonian and reversibility properties of the water wave equations together with Kato's theory of similarity transformations \cite{kato1966} and KAM theory. Reported in these works are explicit expressions for the figure-eight curves, up to real analytic functions of the Floquet exponent $\mu$ and the amplitude of the Stokes waves $\varepsilon$. A low-order approximation of the curves is also given. 

In this work, we obtain high-order asymptotic expansions of the Benjamin-Feir figure-eight curve in finite and infinite depth. In particular, we seek high-order asymptotic estimates for the interval of Floquet exponents parameterizing the figure-eight and for the most unstable eigenvalue. Using the results of \cite{creedonDeconinckTrichtchenko2022}, this allows us to compare the Benjamin-Feir and high-frequency growth rates analytically. This comparison suggests three regimes for Stokes waves: (i) shallow water ($\kappa h < \alpha_{BW}$), in which only high-frequency instabilities exist, (ii) intermediate water ($\alpha_{BW}<\kappa h<\alpha_{DO}(\varepsilon) = 1.4308061674...+\mathcal{O}\left(\varepsilon^2\right)$), in which both instabilities exist but high-frequency instabilities dominate, and (iii) deep water ($\kappa h > \alpha_{DO}(\varepsilon)$), in which both instabilities are present, but the Benjamin-Feir instability dominates.

Our method to obtain these high-order asymptotic approximations is a modification of that developed for high-frequency instabilities in \cite{creedonDeconinckTrichtchenko2021b,creedonDeconinckTrichtchenko2021a,creedonDeconinckTrichtchenko2022}. Although the method is formal, it offers a more direct approach to the Benjamin-Feir figure-eight curve and produces results consistent with numerical computations (for sufficiently small $\varepsilon$) as well as with rigorous results appearing in \cite{bertiMasperoVentura2021,bertiMasperoVentura2022}. The method loses validity for sufficiently large $\varepsilon$, when the Benjamin-Feir instability spectrum separates from the origin and changes its topology, see Figure \ref{fig3}. Some of the lower-order details of our method are also presented by Akers \cite{akers2015} for the Benjamin-Feir instability in infinite depth, although this work uses different conventions for the water wave equations and underlying Stokes waves. In contrast, our expressions are in one-to-one correspondence with those appearing in \cite{bertiMasperoVentura2021,bertiMasperoVentura2022}, giving confidence in the rigorous results as well as in our asymptotic calculations.

\begin{figure}[tb]
    \centering \hspace*{-0.5cm}
    \includegraphics[height=5.9cm,width=13.8cm]{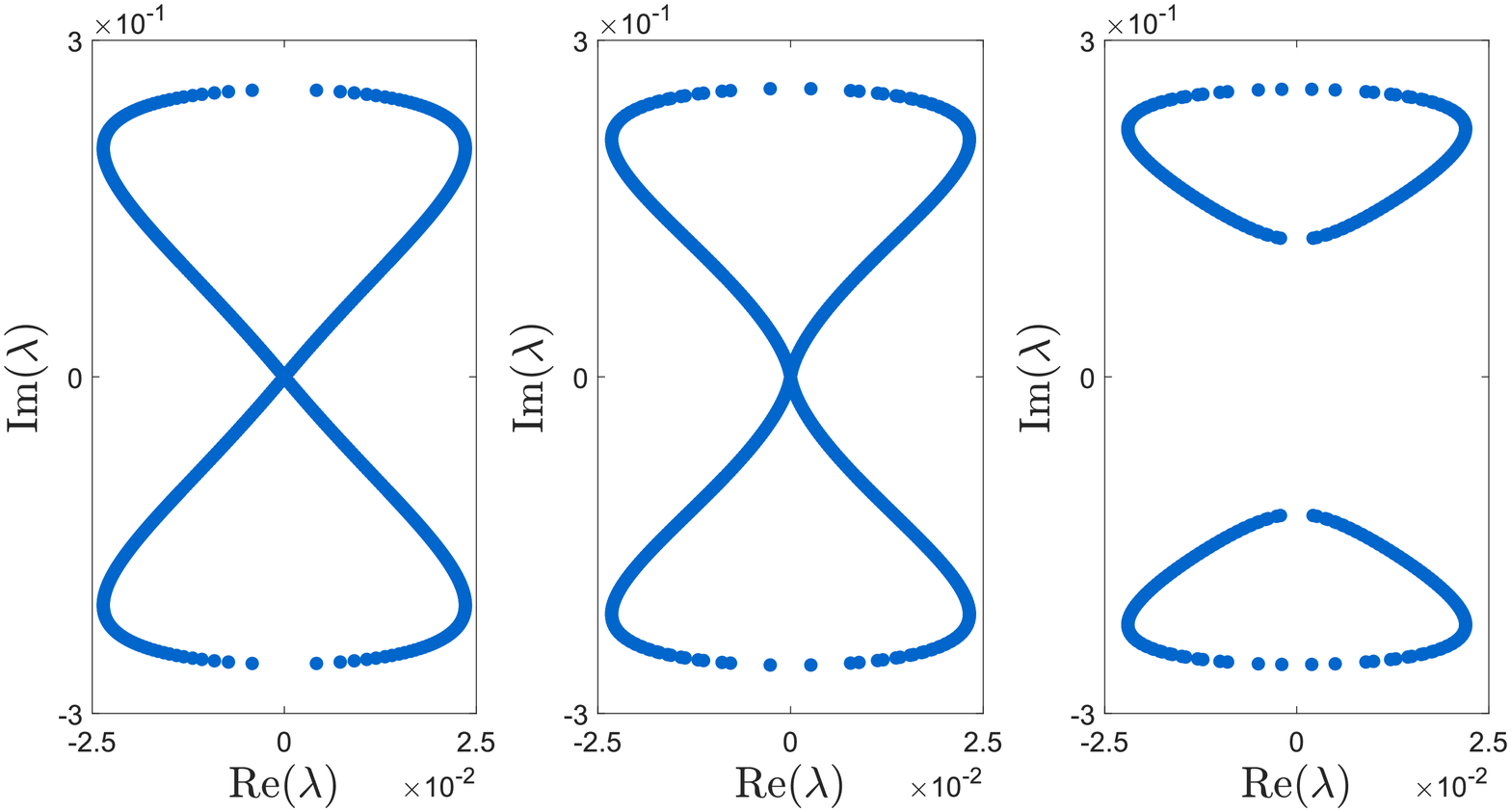}  \vspace*{3mm}
    \caption{\small Numerically computed Benjamin-Feir instability spectra in infinitely deep water for Stokes waves of amplitude $\varepsilon = 0.31$ (left), $\varepsilon = 0.32$ (middle), and $\varepsilon = 0.33$ (right). The methods presented in this work apply only for sufficiently small $\varepsilon$ and, thus, do not capture the separation of the figure-eight from the origin. }
    \label{fig3}
\end{figure}

\section{Preliminaries}

\subsection{The AFM Formulation of the Water Wave Equations}

The Euler equations governing the dynamics of a one-dimensional periodic water wave free of surface tension over an inviscid, irrotational, and two-dimensional bulk are
\begin{subequations}
\begin{align}
\phi_{xx} + \phi_{zz} &= 0, \quad \hspace*{.2cm} \textrm{in} \quad \{(x,z):|x|<\pi/\kappa ~\textrm{and} -h < z < \eta \},  \\
    \eta_t + \eta_x \phi_x &= \phi_z, \quad \textrm{on} \quad z=\eta, \\ 
    \phi_t + \tfrac{1}{2} \left(\phi_x^2 + \phi_z^2\right) + g\eta &= 0, \quad \hspace*{.2cm} \textrm{on} \quad z=\eta, \\
    \phi_z &= 0, \quad \hspace*{0.07cm} \textrm{on} \quad z=-h, 
    \end{align} \label{1.1a}
\end{subequations}
\no together with the periodicity conditions
\begin{subequations}
\begin{align}
    \eta(-\pi/\kappa,t) &= \eta(\pi/\kappa,t), \label{2.1a} \\ \quad \phi_x(-\pi/\kappa,z,t) = \phi_x(\pi/\kappa,z,t), &\quad \phi_z(-\pi/\kappa,z,t) = \phi_z(\pi/\kappa,z,t).
\end{align}
\label{1.1b}
\end{subequations}
\no \hspace*{-0.23cm} Here $x$ is the horizontal coordinate, $z$ is the vertical coordinate, $t$ is time, $\eta(x,t)$ is the surface profile of the water wave relative to the $z=0$ reference level, $\phi(x,z,t)$ is the velocity potential in the bulk, $g$ is the magnitude of the vertical acceleration due to gravity, $h$ is the average depth of the water (assuming constant bathymetry), and $2\pi/\kappa$ is the period of the wave. 

Many formulations of water waves exist that avoid the bulk variable $\phi$ but whose dynamics are fully equivalent to \eqref{1.1a}-\eqref{1.1b}, \emph{e.g.}, \cite{craigSulem1993,dyachenkoZahharovKuznetsov1996,shaw1979,zakharov1968}. We use the Ablowitz-Fokas-Musslimani (AFM) formulation \cite{ablowitzFokasMusslimani2006,ablowitzHaut2008}, which only involves the surface variables $\eta$ and $q = \phi(x,\eta,t)$ and avoids computations of the Dirichlet-to-Neumann operator. The AFM equations can be written as 
\begin{subequations}
\begin{align}
    \int_{-\pi/\kappa}^{\pi/\kappa} e^{-i\kappa m x}\Big[\eta_t\cosh\left(\kappa m\left(\eta + h\right) \right) +iq_x\sinh\left(\kappa m\left(\eta + h \right) \right) \Big] dx &= 0, \quad m \in \mathbb{Z} \setminus \{0\}, \label{1a} \\
    q_t + \frac12 q_x^2 +g\eta - \frac{1}{2}\frac{\left(\eta_t +\eta_xq_x\right)^2}{1+\eta_x^2} &= 0. \label{1b}
\end{align}
\end{subequations}
\no Throughout this work, we refer to \eqref{1a} and \eqref{1b} as the nonlocal and local equations, respectively. 

Moving to a traveling frame $x \rightarrow x -ct$ with velocity $c$ and rescaling variables as in  \cite{creedonDeconinckTrichtchenko2022}, one arrives at the nondimensional AFM equations
\begin{subequations}
\begin{align}
    \int_{-\pi}^{\pi} e^{-i m x}\Big[\left(\eta_t-c\eta_x\right)\cosh\left( m\left(\eta + \alpha\right) \right) +iq_x\sinh\left(m\left(\eta + \alpha \right) \right) \Big] dx &= 0, \quad m \in \mathbb{Z} \setminus \{0\},\label{2a} \\
    q_t - cq_x + \frac12 q_x^2 +\eta - \frac{1}{2}\frac{\left(\eta_t-c\eta_x +\eta_xq_x\right)^2}{1+\eta_x^2} &= 0, \label{2b}
\end{align}
\end{subequations}
\noindent where $\alpha = \kappa h$ is the aspect ratio of the depth of the water to the period of the surface waves modulo $2\pi$. These nondimensionalized equations are equivalent to the dimensional equations (\eqref{1a}-\eqref{1b}) in the traveling frame with $\kappa = 1$ and $g=1$. 

\begin{remark}
In infinitely deep water, the local equation \eqref{2a} is unchanged while the nonlocal equation \eqref{2b} becomes
\begin{align}
    \int_{-\pi}^{\pi} e^{-i m x +|m|\eta} \Big[|m|\left(\eta_t-c\eta_x\right) + i m q_x \Big] dx = 0, \quad m \in \mathbb{Z} \setminus \{0\}.
\end{align}
\end{remark}

\subsection{Small-Amplitude Stokes Waves}

We seek time-independent solutions $\eta_S = \eta_S(x)$ and $q_S= q_S(x)$ of \eqref{2a}-\eqref{2b} that satisfy

\begin{enumerate}

    \item[(i)] $\eta_S$ and $q_S$ are infinitely differentiable with respect to $x$, 
    
    \item[(ii)] $\eta_S$ and $q_{S,x}$ are $2\pi$-periodic with respect to $x$ (but not necessarily $q_S$), 
    
    \item[(iii)] $\eta_S$, $q_{S,x}$, and $c$ depend analytically on a small parameter $\varepsilon$ such that 
    \begin{align} 
    \eta_S\big|_{\varepsilon =0} = 0 = q_{S,x}\big|_{\varepsilon =0}  \quad \textrm{and} \quad ||\eta_S||_{\textrm{L}^2} = \varepsilon + \mathcal{O}\left(\varepsilon^2\right) \quad \textrm{as} \quad \varepsilon \rightarrow 0, \nonumber 
    \end{align} 
    \item[(iv)] $\eta_S$ and $q_{S,x}$ are even in $x$ without loss of generality, and $c(\varepsilon)$ is even in $\varepsilon$,  
    
    \item[(v)] $\eta_S$ has zero average.

\end{enumerate}

\noindent These solutions correspond to the small-amplitude Stokes waves discussed in the Introduction. Their existence is proven in \cite{levicivita1925,nekrasov1921,struik1926}. Expansions of $\eta_S$, $q_S$, and $c$ as power series in $\varepsilon$ are carried out in \cite{creedonDeconinckTrichtchenko2022} using the AFM equations \eqref{2a}-\eqref{2b}. These expansions take the form
\begin{subequations}
\begin{align}
    \eta_S(x;\varepsilon) &= \varepsilon\cos(x) + \sum_{j=2}^{\infty} \eta_j(x)\varepsilon^j, \\
    q_S(x;\varepsilon) &= \frac{\varepsilon}{c_0}\sin(x)+ \sum_{j=2}^{\infty} q_j(x)\varepsilon^j, \\
    c(\varepsilon) &= c_0 + \sum_{j=1}^{\infty} c_{2j}\varepsilon^{2j}, 
\end{align}
\end{subequations}
where $c_0^2 = \tanh(\alpha)$. In what follows, we study right-traveling waves so that $c_0=\sqrt{\tanh(\alpha)}>0$. 

The higher-order corrections of $\eta_S$ and $q_S$ (for $j>1$) take the form
\begingroup
\allowdisplaybreaks
\begin{subequations}
\begin{align}
    \eta_{j}(x) &= \begin{cases}  \quad \displaystyle \sum_{\substack{\ell = 2 \\\ell~\textrm{even}}}^{j}\hat{N}_{j,\ell}\cos(\ell x),  &\textrm{for} ~ j ~\textrm{even}, \\  \quad\displaystyle \sum_{\substack{\ell = 3 \\\ell~\textrm{odd}}}^{j}\hat{N}_{j,\ell}\cos(\ell x), &\textrm{for}~ j ~\textrm{odd}, \end{cases} \label{3a} \\
     q_{j}(x) &= \begin{cases} \quad \displaystyle \hat{Q}_{j,0}x+\sum_{\substack{\ell = 2 \\\ell~\textrm{even}}}^{j}\hat{Q}_{j,\ell}\sin(\ell x), &\textrm{for}~j ~\textrm{even}, \\ \quad \displaystyle \sum_{\substack{\ell = 3 \\\ell~\textrm{odd}}}^{j}\hat{Q}_{j,\ell}\sin(\ell x), &\textrm{for} ~j~ \textrm{odd}, \end{cases} \label{3b}
\end{align}
\end{subequations}
\endgroup
\no where $\hat{N}_{j,\ell}$ and $\hat{Q}_{j,\ell}$ depend only on the aspect ratio $\alpha$. Explicit representations of these coefficients as well as those for $c$ can be found in \cite{creedonDeconinckTrichtchenko2022} to $\mathcal{O}\left(\varepsilon^4\right)$. Taking their limits as $\alpha \rightarrow \infty$ gives the corresponding coefficients in infinite depth.
\begin{remark}
The expansion of $q_S$ contains terms proportional to $x$ due to the induced mean flow of the traveling frame, so $q_S$ is not periodic in general. Despite this, $q_{S,x}$ is periodic, which implies that the horizontal velocity at the surface is periodic.
\end{remark}

\subsection{The Stability Spectrum of Stokes Waves}

Extending our periodic domain to the whole real line, we perturb the Stokes waves using
\begin{align}
    \begin{pmatrix} \eta(x,t;\varepsilon,\rho) \\ q(x,t;\varepsilon,\rho) \end{pmatrix} = \begin{pmatrix} \eta_S(x;\varepsilon) \\ q_S(x;\varepsilon) \end{pmatrix} + \rho \begin{pmatrix} \eta_\rho(x,t) \\ q_\rho(x,t) \end{pmatrix} + \mathcal{O}\left(\rho^2\right), \label{4}
\end{align}
\no where $|\rho|\ll 1$ is a parameter independent of $\varepsilon$ and $\eta_\rho$ and $q_\rho$ are sufficiently smooth functions of $x$ and $t$ that are bounded over the real line for each $t\geq 0$. If \eqref{4} satisfies the AFM equations to $\mathcal{O}(\rho)$, then, since the Stokes waves are periodic \cite{creedonDeconinckTrichtchenko2022,deconinckOliveras2011}, 
\begin{align}
    \begin{pmatrix}
    \eta_{\rho}(x,t) \\
    q_{\rho}(x,t)
    \end{pmatrix} = e^{\lambda t+i\mu x}\begin{pmatrix}
    \mathcal{N}(x) \\
    \mathcal{Q}(x)
    \end{pmatrix} + c.c., \label{5}
\end{align}
\no where $c.c.$ denotes the complex conjugate of what precedes, $\mu \in (-1/2,1/2]$ is the Floquet exponent, and ${\bf w} = (\mathcal{N},\mathcal{Q})^T$ is a sufficiently smooth, $2\pi$-periodic eigenfunction of the spectral problem
\begin{align}
    \mathcal{L}_{\mu,\varepsilon}{\bf w}(x) = \lambda\mathcal{R}_{\mu,\varepsilon}{\bf w}(x), \label{6}
\end{align}
\no with corresponding eigenvalue $\lambda \in \mathbb{C}$.
In general, the linear operators $\mathcal{L}_{\mu,\varepsilon}$ and $\mathcal{R}_{\mu,\varepsilon}$ depend nonlinearly  on $\varepsilon$ and $\mu$. Explicitly, 
\begin{align}
\mathcal{L}_{\mu,\varepsilon} = \begin{pmatrix} \mathcal{L}_{\mu,\varepsilon}^{(1,1)}& \mathcal{L}_{\mu,\varepsilon}^{(1,2)} \\ \mathcal{L}_{\mu,\varepsilon}^{(2,1)} & \mathcal{L}_{\mu,\varepsilon}^{(2,2)} \end{pmatrix}, \quad  \mathcal{R}_{\mu,\varepsilon} = \begin{pmatrix} \mathcal{R}_{\mu,\varepsilon}^{(1,1)}& 0 \\ \mathcal{R}_{\mu,\varepsilon}^{(2,1)} & 1 \end{pmatrix},
\end{align} 

\vspace*{-0.3cm}

\begin{subequations}
\allowdisplaybreaks
\begin{flalign}
    \mathcal{L}_{\mu,\varepsilon}^{(1,1)}[\mathcal{N}(x)] &=  \sum_{n=-\infty}^{\infty}e^{inx} \mathcal{F}_n \big[c\, \mathcal{C}_{n+\mu}\mathcal{D}_x\mathcal{N}(x) +(n+\mu)\left(c\,\mathcal{S}_{n+\mu}\eta_{S,x} -i\mathcal{C}_{n+\mu}q_{S,x} \right)\mathcal{N}(x)\big], \\
   \mathcal{L}_{\mu,\varepsilon}^{(1,2)}[\mathcal{Q}(x)]  &=  \sum_{n=-\infty}^{\infty}e^{inx}\mathcal{F}_n\big[-i\mathcal{S}_{n+\mu}\mathcal{D}_x\mathcal{Q}(x) \big],\\
    \mathcal{L}_{\mu,\varepsilon}^{(2,1)}[\mathcal{N}(x)] &= \eta_{S,x}{\zeta}^2\mathcal{D}_x\mathcal{N}(x) - \mathcal{N}(x), \\
     \mathcal{L}_{\mu,\varepsilon}^{(2,2)}[\mathcal{Q}(x)] &= -{\zeta}\mathcal{D}_x\mathcal{Q}(x), \\
     \mathcal{R}_{\mu,\varepsilon}^{(1,1)}[\mathcal{N}(x)] &= \sum_{n=-\infty}^{\infty}e^{inx}\mathcal{F}_n\left[\mathcal{C}_{n+\mu}\mathcal{N}(x) \right], \\ \mathcal{R}_{\mu,\varepsilon}^{(2,1)}[\mathcal{N}(x)] &= -\eta_{S,x}{\zeta}\mathcal{N}(x),
\end{flalign}
\end{subequations}
\no where 
\begin{align} \mathcal{C}_k = \cosh(k(\eta_S+\alpha)), \quad \quad \mathcal{S}_k = \sinh(k(\eta_S+\alpha)), \quad \quad \mathcal{D}_x = i\mu+\partial_x, 
\quad \quad {\zeta} =\frac{q_{S,x}-  c}{1+\eta_{S,x}^2},
\end{align} 
\no and
\begin{align}
\mathcal{F}_n[f(x)] = \frac{1}{2\pi}\int_{-\pi}^{\pi} e^{-inx} f(x) dx,
\end{align}
\no for any $f(x) \in \textrm{L}^2_{\textrm{per}}(-\pi,\pi)$. The eigenvalues $\lambda \in \mathbb{C}$ of \eqref{6} correspond to the exponential growth rates of the perturbations $\eta_{\rho}$ and $q_{\rho}$. Consequently, we define the stability spectrum of a Stokes wave with amplitude $\varepsilon$ as the union of these eigenvalues over $\mu \in (-1/2,1/2]$. 

\begin{remark}
The eigenvalues of \eqref{6} for fixed $\mu \in (-1/2,0)$ are conjugate to the eigenvalues with Floquet exponent $-\mu$, so we may restrict $\mu \in [0,1/2]$ without loss of generality.
\end{remark}

\begin{remark}
In infinite depth, 
 \begin{subequations}
 \begin{flalign}
    \mathcal{L}_{\mu,\varepsilon}^{(1,1)}[\mathcal{N}(x)] &= \sum_{n=-\infty}^{\infty}\!\! e^{inx}\mathcal{F}_n\big[e^{|n+\mu|\eta_S}|n+\mu|\big(c\mathcal{D}_x\mathcal{N}(x) \!\!+\!\!\big(c\eta_{S,x}|n+\mu|-i(n+\mu)q_{S,x}\big)\mathcal{N}(x) \big) \big], \\
    \mathcal{L}_{\mu,\varepsilon}^{(1,2)}[\mathcal{Q}(x)] &= \sum_{n=-\infty}^{\infty} e^{inx}\mathcal{F}_n\big[e^{|n+\mu| \eta_S}\big(-i(n+\mu)\mathcal{D}_x\mathcal{Q}(x) \big)\big], \\
    \mathcal{R}_{\mu,\varepsilon}^{(1,1)}[\mathcal{N}(x)] &= \sum_{n=-\infty}^{\infty} e^{inx}\mathcal{F}_n\big[e^{|n+\mu|\eta_S}|n+\mu|\mathcal{N}(x)\big].
 \end{flalign}
 \end{subequations}
\no  All other operators remain the same as above.
\end{remark}

For generic choices of $\varepsilon$ and $\mu$, the spectral problem \eqref{6} can only be solved numerically, see \cite{deconinckOliveras2011}, for example. However, in the special case $\varepsilon = 0$, \eqref{6} reduces to a constant-coefficient problem,
    \begin{align}
    \begin{pmatrix} ic_0(\mu+D)\cosh(\alpha(\mu+D)) & (\mu+D)\sinh(\alpha(\mu+D)) \\ -1 & ic_0(\mu+D) \end{pmatrix} {\bf w_0}(x)& 
   \nonumber \\  & \hspace{-3.5cm}
    = \lambda_0\begin{pmatrix} \cosh(\alpha(\mu+D)) & 0 \\ 0 & 1 \end{pmatrix}{\bf w_0}(x),  \label{7}
\end{align} 
\no with $D=-i\partial_x$. We solve \eqref{7} to find
\begin{align}
   \lambda_{0} = -i\Omega_{\sigma}(\mu+n), \quad \quad \sigma = \pm 1, \quad \quad n \in \mathbb{Z},\label{8}
\end{align}
\no for 
\begin{align}
    \Omega_{\sigma}(k) =-c_0k +\sigma \omega(k), \quad \quad \omega(k) = \textrm{sgn}(k)\sqrt{k\tanh(\alpha k)}. \label{9}
\end{align}
\no Equation (\ref{9}) is the linear dispersion relation of the nondimensional AFM equations in a frame traveling at velocity $c_0$. The parameter $\sigma$ specifies the branch of this dispersion relation. Since both branches are real-valued, all eigenvalues $\lambda_0$ are imaginary, and the zero-amplitude Stokes waves are spectrally stable. 

For almost all $\mu$ and $n$, \eqref{8} is a simple eigenvalue with one-dimensional eigenspace spanned by
\begin{align}
    {\bf w_0}(x) = \begin{pmatrix} 1 \\ \dfrac{-i\sigma}{\omega(\mu+n)} \end{pmatrix}e^{inx}.
\end{align}
\no In order for $\lambda_0$ to enter the right-half plane and give rise to an instability for $0<\varepsilon \ll 1$, it must be non-simple. That is, at least two eigenvalues must collide at $\lambda_0$ for the eigenvalues to leave the imaginary axis \cite{deconinckTrichtchenko2017,mackaySaffman1986}. There are a countably infinite number of pairs $(\mu,n)$ for which this occurs~\cite{creedonDeconinckTrichtchenko2021b}. Most of these pairs correspond to eigenvalues away from the origin in the complex spectral plane that generate high-frequency instabilities for $0<\varepsilon\ll 1$ \cite{creedonDeconinckTrichtchenko2022,deconinckTrichtchenko2017,hurYang2020}. The pairs $(0,-1)$, $(0,0)$, and $(0,1)$, however, correspond to the eigenvalue at the origin. This eigenvalue generates the Benjamin-Feir instability in sufficiently deep water for $0 < \varepsilon \ll 1$ \cite{bertiMasperoVentura2021,bertiMasperoVentura2022,bridgesMielke1995,hurYang2020,nguyenStrauss2020}. 
\begin{remark}
In infinite depth, \eqref{7} becomes
\begin{equation}
    \begin{aligned}
    \begin{pmatrix}ic_0(\mu+D)|\mu+D| & (\mu+D)^2  \\ -1 & ic_0(\mu+D) \end{pmatrix} {\bf w_0}(x) = \lambda_0\begin{pmatrix} |\mu+D| & 0 \\ 0 & 1 \end{pmatrix}{\bf w_0}(x).
\end{aligned}  \label{10}
\end{equation}
\no The corresponding eigenvalues take the form \eqref{8} with
\begin{align}
    \Omega_{\sigma}(k) =-c_0k +\sigma \omega(k), \quad \quad \omega(k) = \textup{sgn}(k)\sqrt{|k|}. 
\end{align}
\end{remark}

\subsection{A Spectral Perturbation Method for the Benjamin-Feir Instability}

To reveal the structure in our perturbation calculations, we denote the zero eigenvalue of the $\varepsilon = 0$ spectrum and its corresponding zero Floquet exponent by $\lambda_0$ and $\mu_0$, respectively. However, since $\lambda_0 = 0$ and $\mu_0 = 0$, one can omit these parameters in what follows. 

The eigenvalue $\lambda_0$ is not simple with algebraic multiplicity 4 and geometric multiplicity 3 \cite{bertiMasperoVentura2021,bertiMasperoVentura2022,bridgesMielke1995,hurYang2020,nguyenStrauss2020}. The corresponding eigenspace is spanned by
\begin{align}
    {\bf w_{0,-1}}(x) = \begin{pmatrix} 1 \\ i/c_0 \end{pmatrix}e^{-ix}, \quad \quad {\bf w_{0,0}}(x) = \begin{pmatrix} 0 \\ 1 \end{pmatrix}, \quad \quad   {\bf w_{0,1}}(x) = \begin{pmatrix} 1 \\ -i/c_0 \end{pmatrix}e^{ix}. \label{11}
\end{align}
\no The most general eigenfunction corresponding to $\lambda_0$ is 
\begin{align}
    {\bf w_0}(x) = \beta_{0,-1}{\bf w_{0,-1}}(x) + \beta_{0,0}{\bf w_{0,0}}(x) + \beta_{0,1}{\bf w_{0,1}}(x), \label{11.5}
\end{align}
\no where $\beta_{0,\nu}$ for $\nu \in \{0,\pm 1\}$ are (for now) arbitrary constants. 

\begin{remark}
The generalized eigenspace is spanned by \eqref{10} together with the generalized eigenvector ${\bf v_{0,0}} = (1,0)^T$. We mention this for completeness, but in practice, we only need \eqref{11} to approximate the unstable eigenvalues corresponding to the Benjamin-Feir instability.
\end{remark}

\begin{remark}
Without loss of generality, one of the constants $\beta_{0,\nu}$ can be set to 1 since ${\bf w}_0$ is unique only up to a nonzero scalar. We retain all three constants in our calculations for reasons that become more clear when we consider infinite depth.
\end{remark}

We turn on the $\varepsilon$ parameter and track the unstable eigenvalues near the origin for $0 < \varepsilon \ll 1$. These eigenvalues trace out a figure-eight curve centered at the origin, as mentioned in the Introduction. To track these eigenvalues and their corresponding eigenfunctions, we formally expand in powers of $\varepsilon$:
\begin{subequations}
\begin{align}
    \lambda(\varepsilon) &= \lambda_0 + \sum_{j=1}^{\infty}\lambda_j \varepsilon^j, \label{12a} \\
    {\bf w}(x;\varepsilon) &= {\bf w_0}(x) + \sum_{j=1}^{\infty}{\bf w_j}(x) \varepsilon^j. \label{12b}
\end{align}
\end{subequations} 
\no As the curve deforms with $\varepsilon$, so too does its parameterization in terms of the Floquet exponent. Hence, as in \cite{creedonDeconinckTrichtchenko2022}, we expand this parameter as well, writing
\begin{align}
    \mu(\varepsilon) = \mu_0+\varepsilon\mu_1\big(1+r(\varepsilon)\big), \quad \quad \textrm{with} \quad \quad r(\varepsilon) = \sum_{j=1}^{\infty}r_j\varepsilon^j, \label{13}
\end{align}
\no where $\mu_1$ assumes an interval of values symmetric about zero and $r(\varepsilon)$ captures the higher-order deformations of this interval, see Figure \ref{fig4}.

\begin{remark}
For sufficiently large $\varepsilon$, the figure-eight curve separates from the origin (Figure \ref{fig3}), and thus, the parameterizing interval of Floquet exponents separates into two disjoint intervals. Ansatz \eqref{13} cannot account for this effect, which limits our analysis of the Benjamin-Feir instability spectrum to sufficiently small $\varepsilon$.
\end{remark}

We proceed as follows. Expansions \eqref{12a}, \eqref{12b}, and \eqref{13} are substituted into the full spectral problem \eqref{6}. Powers of $\varepsilon$ are equated, generating a hierarchy of linear inhomogeneous equations for the eigenfunction corrections ${\bf w_j}$. Each of these equations is solvable only if the Fredholm alternative removes secular inhomogeneous terms. This leads to a set of solvability conditions that impose constraints on the eigenvalue corrections $\lambda_j$ as well as corrections to the constants appearing in ${\bf w_0}$. Corrections to the Floquet exponent require an additional constraint called the \emph{regular curve condition} \cite{creedonDeconinckTrichtchenko2021a,creedonDeconinckTrichtchenko2021b,creedonDeconinckTrichtchenko2022}, which ensures the eigenvalue corrections remain bounded as one approaches the intersection of the figure-eight curve with the imaginary axis.

We present this method in more detail in the sections that follow, first in finite depth and then in infinite depth, where more care is needed. Results of the method are compared directly with numerical computations of the Benjamin-Feir instability using methods presented in \cite{deconinckKutz2006,deconinckOliveras2011}. To our knowledge, this is the first time that analytical and numerical descriptions of the Benjamin-Feir figure-eight curve have been quantitatively compared. 

\begin{figure}[tb]
    \centering
    \includegraphics[height=5.2cm,width=7.2cm]{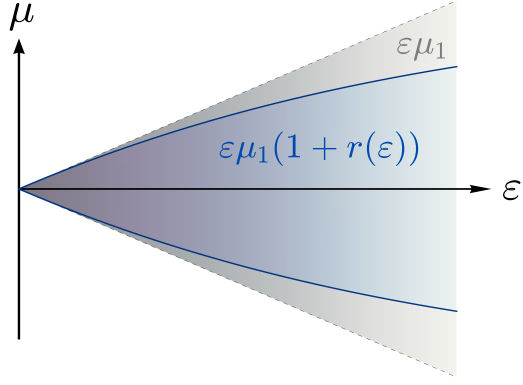} \vspace*{2mm}
    \caption{\small A schematic of the parameterizing interval of Floquet exponents for the Benjamin-Feir figure-eight curve as a function of $\varepsilon$. The gray-shaded region indicates the leading-order approximation of this interval $\varepsilon \mu_1$, where $\mu_1$ is an interval of values symmetric about zero. The blue-shaded region indicates the true interval as a function of $\varepsilon$ and is a uniform rescaling of the leading-order behavior by a factor of $1+r(\varepsilon)$, where $r$ is an analytic function of $\varepsilon$ such that $r(\varepsilon) = o(1)$ as $\varepsilon \rightarrow 0^+$. The boundaries of the true interval may be subtended by curves that are concave up or down, depending on $\alpha$.}
    \label{fig4}
\end{figure}

\section{The Benjamin-Feir Spectrum in Finite Depth}

\subsection{The $\mathcal{O}(\varepsilon)$ Problem}

Substituting \eqref{12a}, \eqref{12b}, and \eqref{13} into the full spectral problem \eqref{6}, terms of $\mathcal{O}(\varepsilon^0)$ necessarily cancel by our choice of $\lambda_0$, ${\bf w_0}$, and $\mu_0$ above. At $\mathcal{O}(\varepsilon)$, we find 
\begin{align}
    \Big(\mathcal{L}_0 - \lambda_0\mathcal{R}_0\Big){\bf w_1} = -\mathcal{L}_1{\bf w_0} + \mathcal{R}_0\Big(\lambda_1{\bf w_0} \Big) + \mathcal{R}_1\Big(\lambda_0{\bf w_0} \Big), \label{14}
\end{align}
\no with
\begin{align}
    \mathcal{L}_j = \frac{1}{j!}\frac{\partial^j}{\partial \varepsilon^j} \mathcal{L}_{\mu(\varepsilon),\varepsilon}, \quad \quad \mathcal{R}_j = \frac{1}{j!}\frac{\partial^j}{\partial \varepsilon^j} \mathcal{R}_{\mu(\varepsilon),\varepsilon}. 
\end{align}
\no The operator $\mathcal{L}_0-\lambda_0\mathcal{R}_0$ is not invertible for $(\lambda_0, \mu_0)=(0,0)$. A solution ${\bf w_1}$ of 
\eqref{14} exists only if the inhomogeneous terms are orthogonal to the nullspace of the adjoint of $\mathcal{L}_0-\lambda_{0}\mathcal{R}_0$ by the Fredholm alternative. A direct calculation shows
\begin{align}
    \textrm{Null}\left(\left(\mathcal{L}_0-\lambda_{0}\mathcal{R}_0\right)^\dagger\right) = \textrm{Span}\left\{ \begin{pmatrix} 1 \\ ic_0 \end{pmatrix} e^{-ix}, \begin{pmatrix} 1 \\ 0 \end{pmatrix},\begin{pmatrix} 1 \\ -ic_0 \end{pmatrix} e^{ix}  \right\}, \label{15}
\end{align}
\no where $\left(\mathcal{L}_0-\lambda_{0}\mathcal{R}_0\right)^{\dagger}$ denotes the adjoint operator with respect to the standard complex inner-product $\left<\cdot,\cdot\right>$ on $\textrm{L}^2_{\textrm{per}}(-\pi,\pi)\times \textrm{L}^2_{\textrm{per}}(-\pi,\pi)$. From \eqref{15}, we arrive at three solvability conditions for \eqref{14}: 
\begin{subequations}
\begin{align}
    \left<  -\mathcal{L}_1{\bf w_0} + \mathcal{R}_0\Big(\lambda_1{\bf w_0} \Big) + \mathcal{R}_1\Big(\lambda_0{\bf w_0} \Big),\begin{pmatrix} 1 \\ ic_0 \end{pmatrix} e^{-ix}\right> &= 0, \label{16a} \\
    \left<  -\mathcal{L}_1{\bf w_0} + \mathcal{R}_0\Big(\lambda_1{\bf w_0} \Big) + \mathcal{R}_1\Big(\lambda_0{\bf w_0} \Big),\begin{pmatrix} 1 \\ 0 \end{pmatrix}\right> &= 0, \label{16b} \\
    \left< -\mathcal{L}_1{\bf w_0} + \mathcal{R}_0\Big(\lambda_1{\bf w_0} \Big) + \mathcal{R}_1\Big(\lambda_0{\bf w_0} \Big),\begin{pmatrix} 1 \\ -ic_0 \end{pmatrix} e^{ix} \right> &= 0. \label{16c}
\end{align}
\end{subequations}
\no Simplifying \eqref{16a} and \eqref{16c} leads to
\begin{subequations}
\begin{align}
    2\beta_{0,-1}\left(\lambda_1 + i\mu_1c_g\right) &= 0, \label{17a} \\
    2\beta_{0,1}\left(\lambda_1+i\mu_1c_g\right) &= 0,\label{17b}
\end{align}
\end{subequations}
\no respectively, where $c_g$ denotes the group velocity of $\Omega_{1}$ \eqref{9} evaluated at $k=1$. Explicitly,
\begin{align}
c_g &= \frac{\alpha(1-c_0^4)-c_0^2}{2c_0}.
\end{align}
\no In contrast to \eqref{16a} and \eqref{16c}, \eqref{16b} reduces to a trivial equality and does not contribute an additional solvability condition.
\begin{remark}
If $c_0>0$, a direct calculation shows that $c_g<0$ for $\alpha>0$. 
\end{remark}
If we require $\beta_{0,\nu} \neq 0$ so that the eigenspace of $\lambda_0$ remains three-dimensional, \eqref{17a} and \eqref{17b} imply
\begin{align}
    \lambda_1 = -i\mu_1c_g. \label{18}
\end{align}
\no Since $\mu_1 \in \mathbb{R}$, the unstable eigenvalues of the Benjamin-Feir instability are imaginary to $\mathcal{O}(\varepsilon)$. 
\begin{remark}
If $\beta_{0,\nu} = 0$ for some $\nu$, one recovers the imaginary spectrum near the origin, as opposed to the figure-eight curve. 
\end{remark}

Before proceeding to $\mathcal{O}\left(\varepsilon^2 \right)$, we solve \eqref{14} subject to \eqref{18}. The solution ${\bf w_1}$ decomposes into a direct sum of a particular solution ${\bf w_{1,p}}$ and a homogeneous solution ${\bf w_{1,h}}$:
\begin{align}
    {\bf w_1}(x) = {\bf w_{1,p}}(x)+{\bf w_{1,h}}(x).
\end{align}
\no The particular solution can be written as
\begin{align}
    {\bf w_{1,p}}(x) = \sum_{j=-2}^{2} {\bf w_{1,j}}e^{ijx},
\end{align}
\no where ${\bf w_{1,j}} = {\bf w_{1,j}}(\alpha,\beta_{0,\nu},\mu_1)\in \mathbb{C}^2$. The homogeneous solution is
\begin{align}
    {\bf w_{1,h}}(x) = \beta_{1,-1}{\bf w_{0,-1}}(x) + \beta_{1,0}{\bf w_{0,0}}(x) + \beta_{1,1}{\bf w_{0,1}}(x),
\end{align}
\no coinciding with the eigenspace of $\lambda_0$. The coefficients $\beta_{1,\nu}$ represent first-order corrections to the zeroth-order eigenfunction correction ${\bf w}_0$ and are undetermined constants at this order.

\begin{remark}
The expressions for ${\bf w_{1,j}}$ as well as for all other algebraic expressions that are too cumbersome to include explicitly in this manuscript are found in our companion Mathematica files. See the Data Availability Statement at the end of this manuscript for access to these files. 
\end{remark}

\subsection{The $\mathcal{O}\left(\varepsilon^2\right)$ Problem}

At $\mathcal{O}\left(\varepsilon^2\right)$, the spectral problem \eqref{6} takes the form
\begin{align}
     \Big(\mathcal{L}_0 - \lambda_0\mathcal{R}_0\Big){\bf w_2} = -\sum_{j=1}^{2}\mathcal{L}_j{\bf w_{2-j}} + \mathcal{R}_0\Big(\sum_{k=1}^{2}\lambda_k{\bf w_{2-k}} \Big) + \sum_{j=1}^{2}\mathcal{R}_j\Big(\sum_{k=0}^{2-j} \lambda_{k}{\bf w_{2-j-k}} \Big). \label{19}
\end{align}
\no Proceeding as above, we obtain three nontrivial solvability conditions for \eqref{19}:
\begin{subequations}
\begin{align}
    2\beta_{0,-1}\big(\lambda_2 + ic_gr_1\mu_1 \big) + \beta_{0,0}S_{2,-1}\mu_1 + i\Big(U_{2,-1}\beta_{0,1} + \big(T_{2,-1}\mu_1^2 + V_{2,-1} \big)\beta_{0,-1} \Big) &= 0, \label{20a} \\
    \beta_{0,0}T_{2,0}\mu_1^2 + iS_{2,0}\mu_1\big(\beta_{0,-1}+\beta_{0,1} \big) &= 0 \label{20b}, \\
    2\beta_{0,1}\big(\lambda_2+ic_gr_1\mu_1\big)+\beta_{0,0}S_{2,1}\mu_1 + i\Big(U_{2,1}\beta_{0,-1} + \big(T_{2,1}\mu_1^2+V_{2,1}\big)\beta_{0,1}\Big) &= 0 \label{20c},
\end{align} 
\end{subequations}
\no where the subscripted coefficients $S$, $T$, $U$, and $V$ are all real-valued functions of the aspect ratio $\alpha$. More explicitly, we have:
\begin{subequations}
\begin{align}
    S_{2,-1} &= \frac{\alpha+5c_0^2-2\alpha c_0^4-c_0^6+\alpha c_0^8}{4c_0^2}, \label{20.1a} \\
    T_{2,-1} &= \frac{\alpha^2-c_0^2\left(-1+\alpha c_0^2 \right)\left(-2\alpha +c_0^2+3\alpha c_0^4 \right)}{4c_0^3}, \label{20.1b}\\
    U_{2,-1} &= \frac{1-2\hat{N}_{2,2}c_0^2\left(1-3c_0^4\right)-8\hat{Q}_{2,2}c_0^3-4c_0^4+c_0^8}{4c_0^3}, \\
    T_{2,0} &= \frac{\alpha^2-2\alpha c_0^2+\left(1-2\alpha^2 \right)c_0^4 -2\alpha c_0^6+\alpha^2c_0^8}{4c_0^2}, \label{20.1c}
\end{align}
\end{subequations}
\no where $\hat{N}_{2,2}$ and $\hat{Q}_{2,2}$ are the second-order corrections of the Stokes waves, see Subsection 2.2 and Appendix A of \cite{creedonDeconinckTrichtchenko2022} for more details. The remaining coefficients follow from the identities 
\begin{equation}
\begin{aligned}
S_{2,-1}=-S_{2,1}, ~  T_{2,-1}=-T_{2,1}, ~U_{2,-1}=-U_{2,1},
~ V_{2,-1}=U_{2,-1}, ~  S_{2,-1}=c_0S_{2,0},
\end{aligned} \label{21}
\end{equation}
\no which hold for $\alpha > 0$. The proofs of these identities are shown in the companion Mathematica files. 


\begin{theorem} For all $\alpha>0$, we have $S_{2,-1}>0$, $T_{2,-1} >0$, and $T_{2,0}<0$. 
\end{theorem}

\begin{proof}
Substituting $c_0 = \sqrt{\tanh(\alpha)}$ in \eqref{20.1a}, we arrive at
\begin{align}
    T_{2,-1} = \frac18 \csch(\alpha)\sech^3(\alpha)\left(2\alpha+3\sinh(2\alpha)+\sinh(4\alpha)\right),
\end{align}
\no from which $S_{2,-1}>0$ follows immediately for $\alpha >0$.

Doing the same for \eqref{20.1b}, we arrive at
\begin{align}
    T_{2,-1} = \frac{-1-4\alpha^2+8\alpha^2\cosh(2\alpha)+\cosh^2(2\alpha)-4\alpha\sinh(2\alpha)}{16\tanh^{3/2}(\alpha)\cosh^4(\alpha)},
\end{align}
\no after some work. 
Rearranging terms in the numerator,
\begin{align}
    T_{2,-1} = \frac{\left(-1-4\alpha^2+\cosh^2(2\alpha) \right)+4\alpha\cosh(2\alpha)\left(2\alpha-\tanh(2\alpha) \right)}{16\tanh^{3/2}(\alpha)\cosh^4(\alpha)}.
\end{align}
\no Using the Taylor series of $\cosh$, we have $-1-4\alpha^2+\cosh^2(2\alpha)>0$ immediately. Using the well-known bound $\tanh(|k|)<|k|$ for $k \in \mathbb{R}$, we have $2\alpha-\tanh(2\alpha)>0$. It follows that $T_{2,-1}>0$ for $\alpha >0$.

Lastly, for \eqref{20.1c}, 
\begin{align}
    T_{2,0} &= -\frac{1}{64}\csch(\alpha)\sech^3(\alpha)\left(e^{4\alpha}-(1+4\alpha) \right)\left(e^{-4\alpha}-(1-4\alpha) \right),
\end{align}
\no after some work. Using $\exp(k)>1+k$ for $k>0$, we immediately conclude $T_{2,0}<0$ for $\alpha > 0$, as desired.
\end{proof}

Equations \eqref{20a}-\eqref{20c} constitute a nonlinear system for the unknown variables $\lambda_2$ and $\beta_{0,\pm 1}$. The first-order Floquet correction $\mu_1$ and first-order rescaling of the Floquet interval $r_1$ appear as parameters in this system. Because of the symmetry of the Floquet interval corresponding to the Benjamin-Feir figure-eight curve, we consider $\mu_1>0$ without loss of generality, as mentioned before. Also appearing as a parameter in our system is  $\beta_{0,0}$, the coefficient of the zeroth mode of ${\bf w_0}$. Without loss of generality, we normalize the eigenfunction ${\bf w}$ so that $\beta_{0,0}>0$. Under these assumptions, we solve \eqref{20a}-\eqref{20c} for $\lambda_2$. Using the identities listed in \eqref{21} as well as the inequalities in the claim above, we find 
\begin{align}
    \lambda_2 &= \lambda_{2,R}+i\lambda_{2,I},
\end{align}
\no where
\begin{subequations}
\begin{align}
    \lambda_{2,R} &= \pm \frac{\mu_1}{2}\sqrt{T_{2,-1}\left(\frac{2\left(S_{2,-1}S_{2,0}-U_{2,-1}T_{2,0} \right)}{T_{2,0}}-T_{2,-1}\mu_1^2\right)}, \label{22a} \\
    \lambda_{2,I} &= - r_1\mu_1c_g. \label{22b}
\end{align}
\end{subequations}
\no Defining
\begin{subequations}
\begin{align}
    e_{2} = 4T_{2,-1}, \quad &e_{BW} = \frac{S_{2,-1}S_{2,0}-U_{2,-1}T_{2,0}}{T_{2,0}}, \quad \textrm{and} \\ \Delta_{BW} &= \sqrt{e_{2}\left(8e_{BW}-e_{2}\mu_1^2\right)}, \end{align}
 \end{subequations}
\no \eqref{22a} simplifies to
\begin{align}
    \lambda_{2,R} = \pm \frac{\mu_1}{8}\Delta_{BW}. \label{23}
\end{align}
\no For \eqref{23} to be nonzero, we must have $e_{BW}>0$. It is well-known (see, for instance, \cite{bertiMasperoVentura2022}) that $e_{BW}>0$ only if $\alpha > \alpha_{BW} = 1.3627827567...$, the critical threshold for modulational instability originally found in \cite{benjamin1967,whitham1967b}. A plot of $e_{BW}$ as a function of $\alpha$ is shown in Figure \ref{fig5}.

\begin{remark}
The variables $e_2$ and $e_{BW}$ correspond directly to the variables $e_{22}$ and $e_{WB}$ in \cite{bertiMasperoVentura2022}, respectively. Using the expressions for $S_{2,-1}$, $U_{2,-1}$, $S_{2,0}$, and $T_{2,0}$ above, we obtain an explicit representation of $e_{BW}$:
\begin{align}
    e_{BW} &=\frac{1}{\left(-1+8\alpha^2+\cosh(4\alpha)-4\alpha\sinh(4\alpha)\right)\tanh^{3/2}(\alpha)}\biggr(-4+8\alpha^2+8\cosh(2\alpha)\phantom) \nonumber \\
    &\phantom(\quad\quad\quad+5\cosh(4\alpha) +2\alpha\Big(-9\coth(\alpha)+18\alpha\csch^2(2\alpha)-2\sinh(4\alpha)+3\tanh(\alpha) \Big)\biggr).\label{23.1}
\end{align}
\no The root of this expression for $\alpha > 0$ is the critical threshold $\alpha_{BW}$.
\end{remark}

\begin{figure}[tb]
    \centering \hspace*{1cm}
    \includegraphics[height=5.2cm,width=8.4cm]{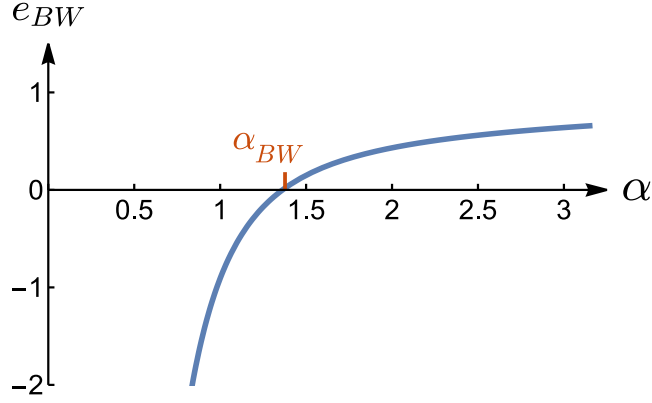} \vspace*{2mm}
    \caption{\small A plot of $e_{BW}$ \emph{vs.} $\alpha$. The only root of $e_{BW}$ for $\alpha>0$ is $\alpha_{BW} = 1.3627827567...$.}
    \label{fig5}
\end{figure}

Provided $\alpha>\alpha_{BW}$, \eqref{23} has nonzero real part for
\begin{align}
    0<\mu_1 <M, \quad M =  \sqrt{\frac{8e_{BW}}{e_{2}}}. \label{24}
\end{align}
\no Inequality \eqref{24} together with the first-order eigenvalue correction \eqref{18} and second-order eigenvalue corrections \eqref{22b} and \eqref{23} yield the leading-order parameterization for one loop of the figure-eight curve. Because $c_g<0$ for all $\alpha >0$, this loop is in the upper-half complex plane. The remaining loop is obtained if one repeats the analysis above for $\mu_1<0$. One finds $-M < \mu_1 < 0$ necessarily, so that the parameterizing interval of Floquet exponents for the entire figure-eight curve has the asymptotic expansion
\begin{align}
    \mu \in  \varepsilon\left(-M,M\right)\Big(1+r_1\varepsilon\Big) + \mathcal{O}\left(\varepsilon^3\right). \label{24.5}
\end{align}

\begin{remark}
For the remainder of this work, we restrict to the positive branch of \eqref{23} and, therefore, obtain a parameterization only for the half loop of the figure-eight curve in the first quadrant of the complex plane. By quadrafold symmetry of the stability spectrum \eqref{6}, we can recover a parameterization for the entire figure-eight curve from this half-loop. 
\end{remark}


Both \eqref{22b} and \eqref{24.5} depend on the first-order rescaling parameter $r_1$. This results in ambiguity at $\mathcal{O}\left(\varepsilon^2\right)$ in both the Floquet parameterization and imaginary part of the figure-eight. We show at the next order that $r_1 = 0$ using the regular curve condition. Using this, we can assemble our expansions for the real and imaginary parts of the figure-eight curve
\begin{subequations}
\begin{align}
    \lambda_{R} &= \frac{\mu_1}{8}\Delta_{BW}\varepsilon^2 + \mathcal{O}\left(\varepsilon^3\right), \label{24.1} \\
    \lambda_{I} &= -\mu_1c_g\varepsilon + \mathcal{O}\left(\varepsilon^3\right), \label{24.2}
\end{align}
\end{subequations}
\no respectively. Dropping terms of $\mathcal{O}\left(\varepsilon^3\right)$ and smaller and eliminating the $\mu_1$ dependence, we obtain the algebraic curve 
\begin{align}
    64c_g^4\lambda_R^2=e_{2}\lambda_{I}^2\biggr(8e_{BW}c_g^2\varepsilon^2-e_{2}\lambda_{I}^2 \biggr),
\end{align}
\no which is a lemniscate of Huygens (or Gerono) \cite{huygens1691}. This lemniscate represents a uniformly accurate asymptotic approximation of the Benjamin-Feir figure-eight curve to $\mathcal{O}\left(\varepsilon^2\right)$ and is consistent with the low-order heuristic approximation presented in \cite{bertiMasperoVentura2022}. For sufficiently small $\varepsilon$, this lemniscate agrees well with numerical results, see Figure \ref{fig6}.

\begin{figure}[tb]
    \centering \hspace*{-0.0cm}
    \includegraphics[height=7.2cm,width=15cm]{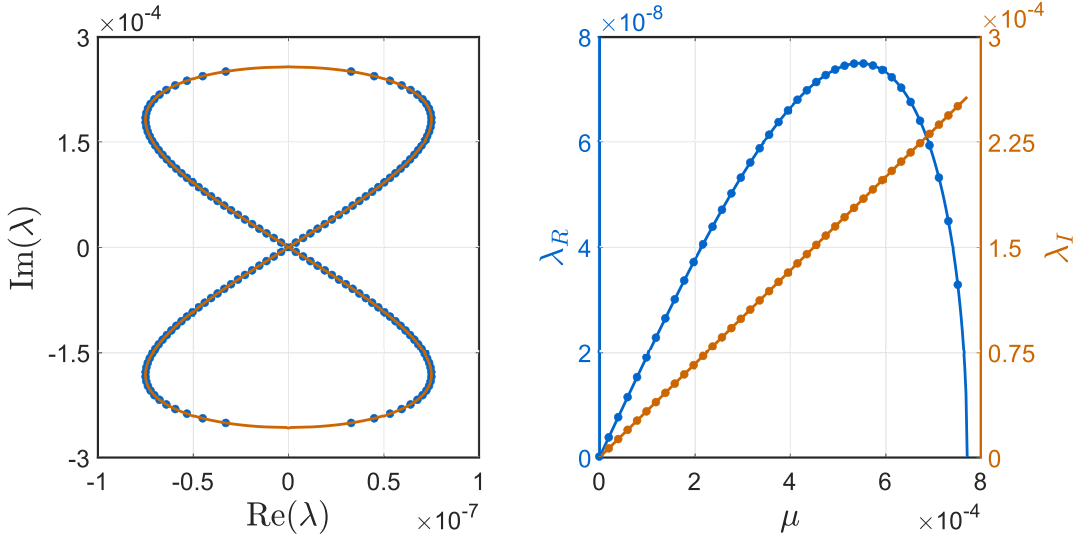} \vspace*{2mm}
    \caption{\small (Left) A plot of the Benjamin-Feir figure-eight curve for a Stokes wave with amplitude $\varepsilon = 10^{-3}$ and aspect ratio $\alpha = 1.5$. Numerical results are given by the blue dots, and the asymptotic results to $\mathcal{O}\left(\varepsilon^2\right)$ are given by the solid orange curve. (Right) The Floquet parameterization of the real (blue axis) and imaginary (orange axis) part of the figure-eight curve on the left. The respective numerical results are given by the correspondingly colored dots, and the respective asymptotic results to $\mathcal{O}\left(\varepsilon^2\right)$ are given by the correspondingly colored curves.}    \label{fig6}
\end{figure}

Given the asymptotic expansion of $\lambda_R$ in \eqref{24.1} above, a direct calculation shows that $\lambda_{R}$ attains the maximum value 
\begin{align}
    \lambda_{R,*} = \frac{e_{BW}}{2}\varepsilon^2 + \mathcal{O}\left(\varepsilon^3\right),
\end{align}
\no when $\mu_1$ is equal to
\begin{align}
    \mu_{1,*} = 2\sqrt{\frac{e_{BW}}{e_{2}}}. \label{24.75}
\end{align}
\no This gives an asymptotic expansion for the real part of the most unstable eigenvalue on the half-loop. Its corresponding imaginary part and Floquet exponent are
\begin{subequations}
\begin{align}
  \lambda_{I,*} &= -c_g\left(2\sqrt{\frac{e_{BW}}{e_{2}}}\right)\varepsilon+\mathcal{O}\left(\varepsilon^2\right), \label{26b} \\
     \mu_* &= \left(2\sqrt{\frac{e_{BW}}{e_{2}}}\right)\varepsilon + \mathcal{O}\left(\varepsilon^2\right), \label{26a}
\end{align}
\end{subequations}
\no respectively. These expansions agree with numerical computations up to $\varepsilon = 10^{-2}$, see Figure~\ref{fig7}.

\begin{figure}[tb]
    \centering \hspace*{-0.5cm}
    \includegraphics[height=7.2cm,width=15.5cm]{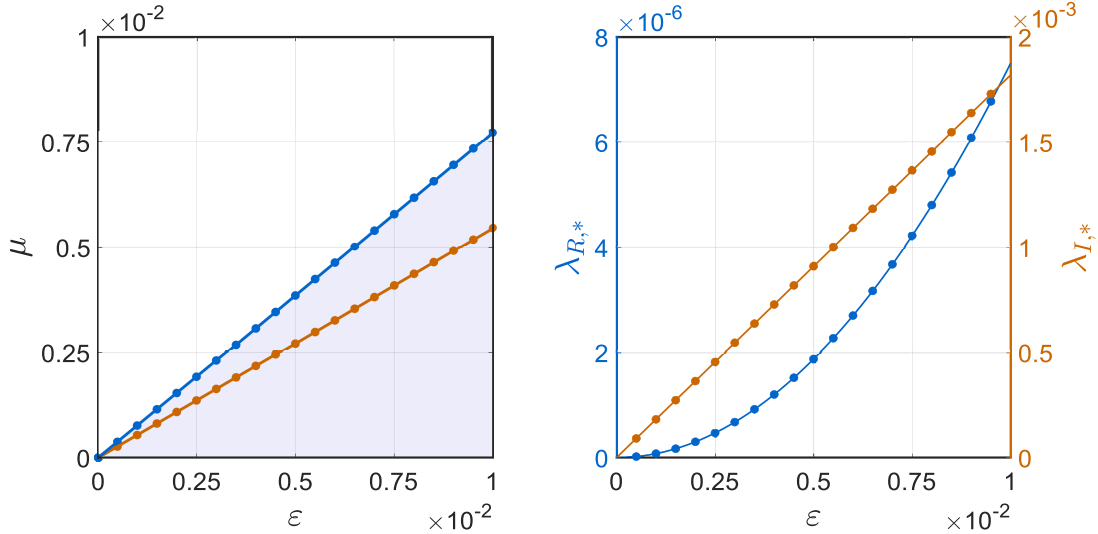} \vspace*{2mm}
    \caption{\small (Left) The interval of Floquet exponents parameterizing the half-loop of the Benjamin-Feir figure-eight curve for a Stokes wave with aspect ratio $\alpha = 1.5$ and variable amplitude $\varepsilon$. The numerically computed boundary of this interval is given by the blue dots, while the solid blue curve gives the asymptotic results to $\mathcal{O}\left(\varepsilon^2\right)$. The orange dots give the numerically computed Floquet exponents of the most unstable eigenvalue, while the solid orange curves give the corresponding asymptotic results to $\mathcal{O}\left(\varepsilon\right)$. (Right) The real (blue axis) and imaginary (orange axis) part of the most unstable eigenvalue with $\alpha = 1.5$ and variable $\varepsilon$. Numerical results are given by the correspondingly colored dots, and the asymptotic results for the real and imaginary part to $\mathcal{O}\left(\varepsilon^2\right)$ and $\mathcal{O}\left(\varepsilon\right)$, respectively, are given by the correspondingly colored solid curves.}    \label{fig7}
\end{figure}

\begin{remark}
As we will see in Subsection 3.4, the first-order Floquet correction $\mu_{1,*}$ corresponding to the most unstable eigenvalue is given by a power series in $\varepsilon$. Equation \eqref{24.75} gives the leading-order term in this series. At this point, we do not know the higher-order corrections of this series and, as a result, are unable to predict the second-order terms of \eqref{26b} and \eqref{26a}. This is a common feature in our analysis of the most unstable eigenvalue: $\lambda_{R,*}$ is determined to one order higher in $\varepsilon$ than $\lambda_{I,*}$ and $\mu_{1,*}$.
\end{remark}

 To conclude our discussion of the $\mathcal{O}\left(\varepsilon^2\right)$ problem, we solve for the remaining unknowns in \eqref{20a}-\eqref{20c} and obtain
\begin{subequations}
\begin{align}
    \beta_{0,-1} &= \frac{\left(ie_{2}\mu_1 \mp \Delta_{BW} \right)T_{2,0}\beta_{0,0}}{2e_{2}S_{2,0}}, \label{26.5a} \\
    \beta_{0,1} &= \frac{\left(ie_{2}\mu_1 \pm \Delta_{BW} \right)T_{2,0}\beta_{0,0}}{2e_{2}S_{2,0}}. \label{26.5b}
\end{align}
\end{subequations}
\no Since we have chosen the positive branch of \eqref{22b}, the negative branch is chosen for \eqref{26.5a}, and the positive branch is chosen for \eqref{26.5b}. Both \eqref{26.5a} and \eqref{26.5b} are determined up to the free parameter $\beta_{0,0}$, which is determined upon choosing a normalization for the eigenfunction ${\bf w}$. 

Finally, given the solutions $\lambda_2$ and $\beta_{0,\pm 1}$ of \eqref{20a}-\eqref{20c}, we solve the second-order problem \eqref{19} and obtain
\begin{align}
    {\bf w_2}(x) = \sum_{j=-3}^{3} {\bf w_{2,j}}e^{ijx}+ \beta_{2,-1}{\bf w_{0,-1}}(x) + \beta_{2,0}{\bf w_{0,0}}(x) + \beta_{2,1}{\bf w_{0,1}}(x),
\end{align}
\no for ${\bf w_{2,j}} = {\bf w_{2,j}}(\alpha,\beta_{0,0},\beta_{1,\nu},\mu_1,r_1)\in \mathbb{C}^2$, see the companion Mathematica files for details. The constants $\beta_{2,\nu} \in \mathbb{C}$ are undetermined at this order.

\subsection{The $\mathcal{O}\left(\varepsilon^3\right)$ Problem}

At $\mathcal{O}\left(\varepsilon^3\right)$, the spectral problem \eqref{6} is
\begin{align}
     \Big(\mathcal{L}_0 - \lambda_0\mathcal{R}_0\Big){\bf w_3} = -\sum_{j=1}^{3}\mathcal{L}_j{\bf w_{3-j}} + \mathcal{R}_0\Big(\sum_{k=1}^{3}\lambda_k{\bf w_{3-k}} \Big) + \sum_{j=1}^{3}\mathcal{R}_j\Big(\sum_{k=0}^{3-j} \lambda_{k}{\bf w_{3-j-k}} \Big). \label{27}
\end{align}
\no Using the solvability conditions from previous orders, the solvability conditions of \eqref{27} simplify to a $3\times 3$ linear system 
\begin{align}
    \mathcal{M} \begin{pmatrix} \beta_{1,-1} \\ \lambda_3 \\ \beta_{1,1} \end{pmatrix} = \begin{pmatrix} f_{3,1} \\ f_{3,2} \\ f_{3,3} \end{pmatrix}, \label{28}
\end{align}
\no with
\begin{align}
     \mathcal{M}&=\begin{pmatrix}
    2(\lambda_2+ic_gr_1\mu_1) +iT_{2,-1}\mu_1^2+iV_{2,-1}& 2\beta_{0,-1} & iU_{2,-1} \\ iS_{2,0}\mu_1 & 0 & iS_{2,0}\mu_1 \\
    iU_{2,1} & 2\beta_{0,1} & 2(\lambda_2+ic_gr_1\mu_1) +iT_{2,1}\mu_1^2+iV_{2,1}
    \end{pmatrix},
\end{align}
\no and
\begingroup
\allowdisplaybreaks
\begin{subequations}
\begin{align}
    f_{3,1} &= -\Big(\beta_{1,0}S_{2,-1}\mu_1 + \beta_{0,-1}\big(2ir_2\mu_1c_g+\mu_1(A_{3,-1}\lambda_2+iB_{3,-1}r_1\mu_1+iC_{3,-1}+iD_{3,-1}\mu_1^2) \big)\phantom) \nonumber \\
    &\phantom(\quad\quad\quad+\beta_{0,0}\big(iE_{3,-1}\lambda_2+F_{3,-1}r_1\mu_1+G_{3,-1}\mu_1^2\big) \Big),  \\
    f_{3,2} &= -\Big(\beta_{1,0}T_{2,0}\mu_1^2+\beta_{0,-1}\big(A_{3,0}\lambda_2+ir_1\mu_1B_{3,0}+iC_{3,0}\mu_1^2 \big)+\mu_1\beta_{0,0}\big(iD_{3,0}\lambda_2+E_{3,0}r_1\mu_1 \big) \phantom) \nonumber \\
    &\phantom(\quad\quad\quad+\beta_{0,1}\big(A_{3,0}\lambda_2+ir_1\mu_1B_{3,0}-iC_{3,0}\mu_1^2 \big)\Big),  \\
    f_{3,3} &= -\Big(\beta_{1,0}S_{2,1}\mu_1 + \beta_{0,1}\big(2ir_2\mu_1c_g+\mu_1(A_{3,1}\lambda_2+iB_{3,1}r_1\mu_1+iC_{3,1}+iD_{3,1}\mu_1^2) \big)\phantom) \nonumber \\
    &\phantom(\quad\quad\quad+\beta_{0,0}\big(iE_{3,1}\lambda_2+F_{3,1}r_1\mu_1+G_{3,1}\mu_1^2\big) \Big).
\end{align}
\end{subequations}
\endgroup
\no The capitalized coefficients in the expressions above are all real-valued functions of the aspect ratio $\alpha$. Explicitly, 
\begingroup
\allowdisplaybreaks
\begin{subequations}
\begin{align}
A_{3,-1} &= 1 + \frac{\alpha}{c_0^2}-\alpha c_0^2, \\
B_{3,-1} &= \frac{\alpha(\alpha-c_0^2+c_0^6-\alpha c_0^8}{c_0^3}, \\
    C_{3,-1} &= \frac{1}{2c_0^5\left(-4c_0^2+\omega_2^2 \right)^2}\Big(8\alpha c_0^{16} + \alpha\omega_2^4 + c_0^2\omega_2^2\left(\alpha-\omega_2^2\right) -c_0^{14}\left( 8+13\alpha\omega_2^2\right)\phantom) \nonumber \\
    &\quad\quad\quad+c_0^6\left(56+17\alpha\omega_2^2+10\omega_2^4\right) + c_0^{12}\left(28\alpha + 22\omega_2^2-4\alpha\omega_2^4\right)-10c_0^4\big(2\alpha \phantom)  \nonumber \\
    &\phantom(\quad\quad\quad \hspace{-0.125cm}+ \omega_2^2+\alpha\omega_2^4\big)+c_0^8\left(-16\alpha-44\omega_2^2+13\alpha\omega_2^4\right)+c_0^{10}\big(16 \phantom) \nonumber \\
    &\phantom(\phantom(\quad\quad\quad\hspace{-0.26cm}-5\omega_2^2\left(\alpha+\omega_2^2\right)\big) - 2c_0^3\left(-\alpha+c_0^2+\alpha c_0^4\right)\left(4c_0^2-\omega_2^2 \right)^2\left(c_2-Q_{2,0}\right)\Big), \\
    D_{3,-1} &= \frac{1}{12 c_0^5}\Big(3\alpha^2\left(\alpha-c_0^2\right)  + \alpha c_0^4(3+\alpha^2)\left(c_0^4-1\right) + c_0^6\left(3+6\alpha^2\right)-3\alpha^2c_0^{10}\left(-1+\alpha c_0^2\right) \Big), \\
    E_{3,-1} &= \frac{1-c_0^4}{2c_0}, \\
    F_{3,-1} &= \frac12\left(3-c_0^4\right), \\
    G_{3,-1} &= \frac{\alpha^2+2\alpha c_0^2 - c_0^4(3+\alpha^2)-c_0^8(\alpha^2-1)-2\alpha c_0^{10}+\alpha^2c_0^{12}}{8c_0^4}, \\
    E_{3,0} &= -\alpha + c_0^2-\alpha c_0^4,
\end{align}
\end{subequations}
\endgroup
\no where $Q_{2,0}$ is a second-order correction of the Stokes wave due to the traveling frame (see Subsection 2.2) and $\omega_2 = \omega(2)$ for $\omega$ in \eqref{9}. Analogous to \eqref{21}, the remaining coefficients are determined by the following identities for $\alpha > 0$:
\begin{align}
  & \hspace*{-0.5cm} A_{3,-1} = -A_{3,1}, ~ B_{3,-1} =-B_{3,1}, C_{3,-1} = C_{3,1}, ~ D_{3,-1} = D_{3,1}, ~E_{3,-1} =-E_{3,1}, ~ F_{3,-1} =-F_{3,1},\nonumber \\
  &  ~ G_{3,-1} = G_{3,1}, ~ A_{3,0} = -E_{3,-1}/c_0, ~B_{3,0} = F_{3,-1}/c_0, ~ C_{3,0} = G_{3,-1}/c_0, ~D_{3,0} = c_0A_{3,-1}.   \label{28.5}
\end{align}
\no In addition, we have a new identity
\begin{align}
    T_{2,0}\Big(S_{2,0}\Big(c_gE_{3,-1}+F_{3,-1}\Big)+S_{2,-1}\Big(-c_gA_{3,0}+B_{3,0}\Big)\Big)-S_{2,-1}S_{2,0}\Big(c_gD_{3,0}+E_{3,0} \Big) = 0, \label{28.55}
\end{align}
\no to be used momentarily. The proofs of \eqref{28.5} and \eqref{28.55} are found in the companion Mathematica files.

Taking the positive branch of $\lambda_2$ and corresponding branches of $\beta_{0,\pm 1}$, a direct calculation shows
\begin{align}
    \textrm{det}(\mathcal{M}) = \beta_{0,0}T_{2,0}\Delta_{BW}\mu_1^3,
\end{align}
\no which is nonzero for $\mu_1$ satisfying \eqref{24}. A similar result holds if the negative branch of $\lambda_2$ is chosen. Thus, \eqref{28} is an invertible linear system for all eigenvalues along the figure-eight curve. Solving this system for $\lambda_3$ on the half-loop, we find
\begin{align}
    \lambda_3 =\lambda_{3,R}+ i\lambda_{3,I},
\end{align}
\no where
\begin{subequations}
\begin{align}
    \lambda_{3,R} &= \frac14r_1\mu_1\biggr(\frac{e_2\Lambda_{3,R}}{\Delta_{BW}}-\frac{\Delta_{BW}\left(c_gA_{3,-1}-B_{3,-1} \right)}{e_{2}}\biggr), \label{29a} \\
    \lambda_{3,I} &= \mu_1\biggr(-r_2c_g + \frac{\Lambda_{3,I}}{32e_2T_{2,0}^2} \biggr), \label{29b}
\end{align} 
\end{subequations}
\no and, after using \eqref{28.55} to simplify,
\begin{subequations}
\begin{align}
    \Lambda_{3,R} &= \mu_1^2\left(c_gA_{3,-1}-B_{3,-1}\right),  \\
    \Lambda_{3,I} &= -A_{3,-1}\Delta_{BW}^2T_{2,0}^2-16e_2T_{2,0}\left(-C_{3,0}S_{2,-1}-G_{3,-1}S_{2,0}+T_{2,0}\left(C_{3,-1}+D_{3,-1} \right) \right) \nonumber \\ &\quad\quad\quad+e_2^2\left(-2D_{3,0}S_{2,-1}S_{2,0}+T_{2,0}\left(-2A_{3,0}S_{2,-1}+2E_{3,-1}S_{2,0}+A_{3,-1}T_{2,0}\mu_1^2 \right) \right). 
\end{align}
\end{subequations} 
\no Solutions of \eqref{27} for $\beta_{1,\pm 1}$ are found in the companion Mathematica files. 

It appears \eqref{29a} is singular as $\mu_1 \rightarrow M$ since $\Delta_{BW} \rightarrow 0$. If $r_1\neq 0$, this singularity is not removable, as the following proposition shows. 

\begin{theorem}
Let $\Lambda_{3,R}^{(M)} = \lim_{\mu_1\rightarrow M} \Lambda_{3,R}$. For $\alpha>\alpha_{BW}$, $\Lambda_{3,R}^{(M)}\neq 0$. 
\end{theorem}

\begin{proof}
Taking the appropriate limit of $\Lambda_{3,R}$ yields
\begin{align}
    \Lambda_{3,R}^{(M)} = \frac{8e_{BW}}{e_2}\Big(c_gA_{3,-1}-B_{3,-1}\Big). 
\end{align}
\no Using explicit expressions for $A_{3,-1}$, $B_{3,-1}$, and $c_g$, a direct calculation shows
\begin{align}
   c_gA_{3,-1}-B_{3,-1} =-2T_{2,-1}.
\end{align}
\no Given $e_2 = 4T_{2,-1}$ by definition, we conclude
\begin{align}
    \Lambda_{3,R}^{(M)} = -4e_{BW} <0, 
\end{align}
\no for $\alpha > \alpha_{BW}$. This proves the claim.
\end{proof}

Because $\Lambda_{3,R}^{(M)}$ is nonzero for all $\alpha>\alpha_{BW}$, \eqref{29a} is singular as $\mu_1 \rightarrow M$, unless $r_1=0$. Since the Benjamin-Feir figure-eight curve consists of bounded eigenvalues that have non-singular dependence on the Floquet exponent \cite{bertiMasperoVentura2022}, the regular curve condition \cite{creedonDeconinckTrichtchenko2021b,creedonDeconinckTrichtchenko2021a,creedonDeconinckTrichtchenko2022} enforces the choice $r_1 = 0$ to remove the singularity, justifying our claim at the previous order.
 
With $r_1 = 0$, $\lambda_3$ is imaginary and depends on the second-order rescaling parameter $r_2$. To determine $r_2$ and the next real correction to the figure-eight curve, we must proceed to $\mathcal{O}\left(\varepsilon^4\right)$.
This requires the solution of \eqref{27} subject to the solvability conditions above. We obtain
\begin{align}
    {\bf w_3}(x) = \sum_{j=-4}^{4} {\bf w_{3,j}}e^{ijx} + \beta_{3,-1}{\bf w_{0,-1}}(x) + \beta_{3,0}{\bf w_{0,0}}(x) + \beta_{3,1}{\bf w_{0,1}}(x),
\end{align}
\no where ${\bf w_{3,j}} = {\bf w_{3,j}}(\alpha,\beta_{0,0},\beta_{1,0},\beta_{2,\nu},\mu_1,r_2)\in \mathbb{C}^2$ while $\beta_{3,\nu} \in \mathbb{C}$ are undetermined constants at this order, see the companion Mathematica files for details.

\subsection{The $\mathcal{O}\left(\varepsilon^4\right)$ Problem}

At $\mathcal{O}\left(\varepsilon^4\right)$, the spectral problem \eqref{6} is
\begin{align}
     \Big(\mathcal{L}_0 - \lambda_0\mathcal{R}_0\Big){\bf w_4} = -\sum_{j=1}^{4}\mathcal{L}_j{\bf w_{4-j}} + \mathcal{R}_0\Big(\sum_{k=1}^{4}\lambda_k{\bf w_{4-k}} \Big) + \sum_{j=1}^{4}\mathcal{R}_j\Big(\sum_{k=0}^{4-j} \lambda_{k}{\bf w_{4-j-k}} \Big). \label{30}
\end{align}
\no The solvability conditions of \eqref{30} simplify to a $3\times 3$ linear system 
\begin{align}
    \mathcal{M} \begin{pmatrix} \beta_{2,-1} \\ \lambda_4 \\ \beta_{2,1} \end{pmatrix} = \begin{pmatrix} f_{4,1} \\ f_{4,2} \\ f_{4,3} \end{pmatrix}, \label{31}
\end{align}
\no where $\mathcal{M}$ is as before and
\begingroup
\allowdisplaybreaks
\begin{subequations}
\begin{align}
    f_{4,1} &= -\Big(\beta_{2,0}S_{2,-1}\mu_1 + \beta_{1,-1}\big(2\left(\lambda_3+ir_2\mu_1c_g\right)+\mu_1(A_{3,-1}\lambda_2+iC_{3,-1}+iD_{3,-1}\mu_1^2) \big)\phantom) \nonumber \\
    &\phantom(\quad\quad\quad+\beta_{1,0}\big(iE_{3,-1}\lambda_2+G_{3,-1}\mu_1^2\big) +\beta_{0,-1}\big(2ir_3\mu_1c_g +\mu_1(A_{3,-1}\lambda_3+iB_{3,-1}r_2\mu_1 \phantom) \phantom) \phantom) \nonumber \\
    &\phantom(\phantom(\phantom(\quad\quad\quad\hspace*{-0.28cm}+ i\mu_1^3I_{4,-1} + \mu_1(H_{4,-1}\lambda_2 + iG_{4,-1}^{0,-1}) )\big) + \lambda_2E_{4,-1} + iJ_{4,-1}^{0,-1} - i\lambda_2^2/c_0\big) \nonumber \\
      &\phantom(\quad\quad\quad+\beta_{0,0}\big(iE_{3,-1}\lambda_3+F_{3,-1}r_2\mu_1+D_{4,-1}\mu_1^3+\mu_1\left(iA_{4,-1}\lambda_2+C_{4,-1}\right) \big)\phantom)  \nonumber \\ &\phantom(\quad\quad\quad+ i\beta_{0,1} \big(\mu_1^2G_{4,-1}^{0,1} +J_{4,-1}^{0,1}\big) \Big),  \\
         f_{4,2} &= -\Big(\beta_{2,0}T_{2,0}\mu_1^2+\beta_{1,-1}\big(A_{3,0}\lambda_2+iC_{3,0}\mu_1^2 \big)+i\beta_{1,0}D_{3,0}\mu_1\lambda_2 \phantom) \nonumber \\
   &\phantom(\quad\quad\quad+\beta_{1,1}\big(A_{3,0}\lambda_2-iC_{3,0}\mu_1^2 \big) + \beta_{0,-1}\big(A_{3,0}\lambda_3 + iB_{3,0}r_2\mu_1 + \mu_1(D_{4,0}\lambda_2+iF_{4,0} \phantom)\phantom)\phantom) \nonumber \\
    &\phantom(\phantom(\phantom(\quad\quad\quad\hspace*{-0.28cm}+iG_{4,0}\mu_1^2) \big) +\beta_{0,0}\big(\mu_1(iD_{3,0}\lambda_3+E_{3,0}r_2\mu_1+H_{4,0}\mu_1^3 + C_{4,0}\mu_1)-\lambda_2^2 \big) \phantom) \nonumber \\ &\phantom(\quad\quad\quad+\beta_{0,1}\big(A_{3,0}\lambda_3+iB_{3,0}r_2\mu_1+\mu_1(-D_{4,0}\lambda_2+iF_{4,0}+iG_{4,0}\mu_1^2) \big) \Big),
 \\
 f_{4,3} &= -\Big(\beta_{2,0}S_{2,1}\mu_1 + \beta_{1,1}\big(2\left(\lambda_3+ir_2\mu_1c_g\right)+\mu_1(A_{3,1}\lambda_2+iC_{3,1}+iD_{3,1}\mu_1^2) \big)\phantom) \nonumber \\
    &\phantom(\quad\quad\quad+\beta_{1,0}\big(iE_{3,1}\lambda_2+G_{3,1}\mu_1^2\big) +\beta_{0,1}\big(2ir_3\mu_1c_g +\mu_1(A_{3,1}\lambda_3+iB_{3,1}r_2\mu_1 \phantom) \phantom) \phantom) \nonumber \\
    &\phantom(\phantom(\phantom(\quad\quad\quad\hspace*{-0.28cm}+ i\mu_1^3I_{4,1} + \mu_1(H_{4,1}\lambda_2 + iG_{4,1}^{0,1}))\big) + \lambda_2E_{4,1} + iJ_{4,1}^{0,1} + i\lambda_2^2/c_0\big)  \nonumber \\
    &\phantom(\quad\quad\quad+\beta_{0,0}\big(iE_{3,1}\lambda_3+F_{3,1}r_2\mu_1+D_{4,1}\mu_1^3+\mu_1\left(iA_{4,1}\lambda_2+C_{4,1}\right) \big)\phantom)  \nonumber \\ 
    &\phantom(\quad\quad\quad+ i\beta_{0,-1} \big(\mu_1^2G_{4,1}^{0,-1} +J_{4,1}^{0,-1}\big) \Big). 
\end{align}
\end{subequations}
\endgroup
As before, the capitalized coefficients above are all real-valued functions of $\alpha$. The interested reader can consult the companion Mathematica files for the explicit representations of these functions. One can show that 
\begin{align}
  &\hspace*{-0.5cm} A_{4,-1} = A_{4,1}, ~ C_{4,-1} =-C_{4,1}, ~ D_{4,-1} = -D_{4,1}, E_{4,-1} = E_{4,1}, ~G_{4,-1}^{0,1} =-G_{4,1}^{0,-1}, \nonumber \\
 & \hspace*{-0.5cm} G_{4,-1}^{0,-1} =-G_{4,1}^{0,1}, ~ H_{4,-1} = H_{4,1}, ~ I_{4,-1} = -I_{4,1}, ~ J_{4,-1}^{0,1} = -J_{4,1}^{0,-1}, ~ J_{4,-1}^{0,-1} = -J_{4,1}^{0,1}, \nonumber\\
 & \hspace*{1.5cm} J_{4,-1}^{0,-1} = J_{4,-1}^{0,1}, ~ D_{4,0} = -A_{4,-1}/c_0, ~ F_{4,0} = C_{4,-1}/c_0,
\end{align}
\no for $\alpha>0$, analogous to \eqref{21} and \eqref{28.5} from the previous orders. 
Solving \eqref{31} for $\lambda_4$ on the half loop yields
\begin{align}
    \lambda_4 = \lambda_{4,R} +i\lambda_{4,I},  \label{32}
\end{align}
\no with
\begin{subequations}
\begin{align}
    \lambda_{4,R} &= \frac{\mu_1}{256T_{2,0}^3}\biggr(\frac{\Lambda_{4,R}^{(1)}}{T_{2,0}\Delta_{BW}}-\frac{\Delta_{BW}\Lambda_{4,R}^{(2)}}{c_0e_2^2}\biggr), \label{32.5a}\\
    \lambda_{4,I} &= -r_3\mu_1c_g.\label{32.5b}
\end{align}
\end{subequations}
\no The coefficients $\Lambda_{4,R}^{(j)}$ in \eqref{32.5a} decompose as
\begin{align}
    \Lambda_{4,R}^{(j)} = \Lambda_{4,R}^{(j,1)}r_2 + \Lambda_{4,R}^{(j,2)}, \quad  j \in \{1,2\}.
\end{align}
\no An application of \eqref{28.55} shows
\begin{align}
    \Lambda_{4,R}^{(1,1)} &= 64e_{2}T_{2,0}^4\mu_1^2\left(c_gA_{3,-1}-B_{3,-1} \right), \\
    \Lambda_{4,R}^{(2,1)} &= 64c_0e_2T_{2,0}^2\left(c_gA_{3,-1}-B_{3,-1}\right).
\end{align}
\no The remaining coefficients $\Lambda_{4,R}^{(j,2)}$ are explicit functions of $\alpha$ and $\mu_1^2$. These coefficients as well as the solutions $\beta_{2,\pm 1}$ of \eqref{31} are found in the companion Mathematica files. 

Similar to the previous order, the real part of $\lambda_{4}$ is singular as $\mu_1 \rightarrow M$. To remove this singular behavior, we require
\begin{align}
    r_{2} = -\frac{\Lambda_{4,R}^{(1,2,M)}}{\Lambda_{4,R}^{(1,1,M)}},
\end{align}
\no using the regular curve condition, where
\begin{align}
    \Lambda_{4,R}^{(1,j,M)} = \lim_{\mu_1\rightarrow M}\Lambda_{4,R}^{(1,j)}, \quad  j \in \{1,2\}.
\end{align}
\no The rescaling parameter $r_2$ is well-defined for any fixed $\alpha>\alpha_{BW}$, since $\Lambda_{4,R}^{(1,1,M)} \neq 0$ over this interval by arguments similar to those in Proposition 3.3.1. However, $r_2$ is unbounded as $\alpha \rightarrow \alpha_{BW}^+$ or $\alpha \rightarrow \infty$, see Figure \ref{fig8}. Both limits suggest potentially unbounded growth in the imaginary part of the figure-eight curve at $\mathcal{O}\left(\varepsilon^3\right)$ \eqref{29b} and the real part of the curve at $\mathcal{O}\left(\varepsilon^4\right)$ \eqref{32.5a}. Because $\mu_1$ appears as a factor in both of these expressions, the apparent singular behavior as $\alpha \rightarrow \alpha_{BW}^+$ is arrested since $\mu_1 \rightarrow 0$ in this limit. 
\begin{figure}[tb]
    \centering \hspace*{-0.0cm}
    \includegraphics[height=5.2cm,width=9.8cm]{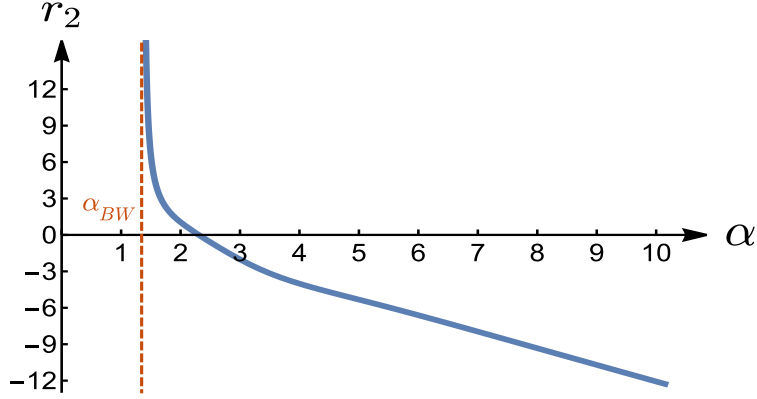} \vspace*{3mm}
    \caption{\small A plot of $r_2$ \emph{vs.} $\alpha$. For all $\alpha > \alpha_{BW}$, $r_2$ is well-defined. As $\alpha \rightarrow \alpha_{BW}^+$ or $\alpha \rightarrow \infty$, $r_2$ becomes singular. The singular behavior as $\alpha \rightarrow \alpha_{BW}^+$ is arrested by the factor of $\mu_1$ in front of \eqref{29b} and \eqref{32.5a}. The singular behavior as $\alpha \rightarrow \infty$ remains, showcasing the breakdown of compactness in finite versus infinite depth, see \cite{bertiMasperoVentura2022,nguyenStrauss2020} for further discussion.}    \label{fig8}
\end{figure}

The same cannot be said as $\alpha \rightarrow \infty$. The culprits for this growth are the expressions for $\beta_{0,\pm1}$ obtained at $\mathcal{O}\left(\varepsilon^2\right)$, see \eqref{26.5a} and \eqref{26.5b}. In particular, $\beta_{0,\pm1}$ both share a factor of $T_{2,0}$ in their respective numerators that becomes unbounded as $\alpha \rightarrow \infty$. This singular behavior is inherited by $r_2$ and ultimately affects the real and imaginary parts of the figure-eight curve at $\mathcal{O}\left(\varepsilon^4\right)$ and $\mathcal{O}\left(\varepsilon^3\right)$, respectively. This provides a first glimpse into the breakdown of compactness that distinguishes the Benjamin-Feir instability spectrum in finite and infinite depth, as discussed in more detail by \cite{bertiMasperoVentura2022,nguyenStrauss2020}. This difference will become even more clear when we consider the infinite depth case in Section 4.

\begin{remark}
The singular behavior of $r_2$ as $\alpha \rightarrow \alpha_{BW}^+$ and $\alpha \rightarrow \infty$ also affects the parameterizing interval of Floquet exponents \eqref{13} at $\mathcal{O}\left(\varepsilon^3\right)$. For similar reasons as above, this singular behavior is avoided as $\alpha\rightarrow \alpha_{BW}^+$, but remains as $\alpha \rightarrow \infty$.
\end{remark}

In general, \eqref{32} has nonzero real part, and we have found a higher-order approximation to the figure-eight curve. This curve is parameterized by Floquet exponents 
\begin{equation}
    \mu \in \varepsilon\left(-M,M\right)\Big(1-\frac{\Lambda_{4,R}^{(1,2,M)}}{\Lambda_{4,R}^{(1,1,M)}}\varepsilon^2 + r_3\varepsilon^3\Big) + \mathcal{O}\left(\varepsilon^5\right). \label{32.25}
\end{equation}
\no The real part along a half-loop of this curve has asymptotic expansion
\begingroup \allowdisplaybreaks
\begin{align}
    \lambda_{R} &= \frac{\mu_1}{8}\Delta_{BW}\varepsilon^2 + \frac{\mu_1}{256c_0\Delta_{BW}\Lambda_{4,R}^{(1,1,M)}e_2^2T_{2,0}^4}\biggr(c_0e_2^2\left(\Lambda_{4,R}^{(1,1,M)}\Lambda_{4,R}^{(1,2)}-\Lambda_{4,R}^{(1,2,M)}\Lambda_{4,R}^{(1,1)} \right)\phantom)  \nonumber \\
    &\phantom(\quad\quad\quad - T_{2,0}\Delta_{BW}^2\left(\Lambda_{4,R}^{(1,1,M)}\Lambda_{4,R}^{(2,2)}-\Lambda_{4,R}^{(1,2,M)}\Lambda_{4,R}^{(2,1)} \right)\biggr)\varepsilon^4 + \mathcal{O}\left(\varepsilon^5\right), \label{33} 
\end{align}
\endgroup
\no for $0<\mu_1<M$, and its corresponding imaginary part is
\begin{equation}
    \lambda_{I} = -\mu_1c_g\varepsilon +\mu_1\biggr( \frac{c_g\Lambda_{4,R}^{(1,2,M)}}{\Lambda_{4,R}^{(1,1,M)}} + \frac{\Lambda_{3,I}}{32e_{2}T_{2,0}^2} \biggr)\varepsilon^3 -r_3\mu_1c_g\varepsilon^4 + \mathcal{O}\left(\varepsilon^5\right). \label{33.5}
\end{equation}
\no Quadrafold symmetry of the stability spectrum \eqref{6} extends \eqref{33} and \eqref{33.5} to a full parameterization of the higher-order approximation of the figure-eight curve.

At this order, $r_3$ is undetermined, leading to ambiguities in the Floquet parameterizing interval \eqref{32.25} and the imaginary part \eqref{33.5}. Proceeding to $\mathcal{O}\left(\varepsilon^5\right)$, one can show via the regular curve condition that $r_3 = 0$. Dropping terms of at least $\mathcal{O}\left(\varepsilon^5\right)$ in \eqref{33} and \eqref{33.5} and eliminating $\mu_1$ leads, in theory, to a new algebraic curve that uniformly approximates the Benjamin-Feir figure-eight to $\mathcal{O}\left(\varepsilon^4\right)$. In practice, eliminating $\mu_1$ from \eqref{33} and \eqref{33.5} is too cumbersome, and we leave this curve in its parameterized form on the half-loop. Figure \ref{fig9} compares our higher-order approximation of the figure-eight with numerical results and the lower-order approximation of the figure-eight, obtained above. Both figure-eight approximations match numerical computations well for $\varepsilon \ll 0.1$. Around $\varepsilon = 0.1$, the lower-order approximation deviates from numerical results, while the higher-order approximation maintains excellent agreement, giving confidence in our higher-order asymptotic expansions.

\begin{figure}[tb]
    \centering \hspace*{-0.0cm}
    \includegraphics[height=7.2cm,width=15cm]{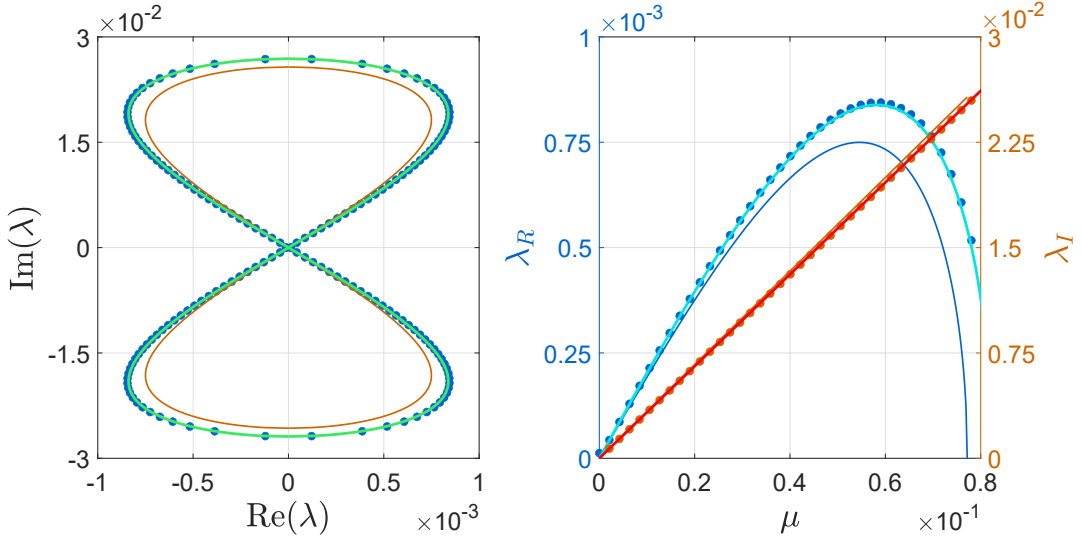} \vspace*{2mm}
    \caption{\small (Left) A plot of the Benjamin-Feir figure-eight curve for a Stokes wave with amplitude $\varepsilon = 0.1$ and aspect ratio $\alpha = 1.5$. Numerical results are given by the blue dots, while asymptotic results to $\mathcal{O}\left(\varepsilon^2\right)$ and $\mathcal{O}\left(\varepsilon^4\right)$ are given by the solid orange and green curves, respectively. (Right) The Floquet parameterization of the real (blue axis) and imaginary (orange axis) part of the figure-eight curve on the left. Numerical results are given by the correspondingly colored dots. The asymptotic parameterizations of the real part to $\mathcal{O}\left(\varepsilon^2\right)$ and $\mathcal{O}
    \left(\varepsilon^4\right)$ are given by the solid blue and light blue curves, respectively, while those for the imaginary part are given by the solid orange and red curves, respectively. }    \label{fig9}
\end{figure}

In addition to a higher-order description of the figure-eight curve, we can estimate its most unstable eigenvalue by examining the critical points of \eqref{33} with respect to $\mu_1$. For ease of notation, let $\lambda_{2,R}$ and $\lambda_{4,R}$ denote the second- and fourth-order corrections of \eqref{33}, respectively, and let $\mu_{1,*}$ denote the critical points. Then
\begin{align}
    \frac{\partial}{\partial \mu_1}\biggr(\lambda_{2,R}(\alpha,\mu_1)\varepsilon^2 + \lambda_{4,R}(\alpha,\mu_1)\varepsilon^4  + \mathcal{O}\left(\varepsilon^5\right) \biggr)\biggr|_{\mu_{1,*}}= 0.
\end{align}
Dropping terms of $\mathcal{O}\left(\varepsilon^5\right)$ and smaller, we arrive at an algebraic equation for the critical points:
\begin{align}
\lambda_{2,R}'(\alpha,\mu_{1,*}) + \lambda_{4,R}'(\alpha,\mu_{1,*})\varepsilon^2 &= 0, \label{34}
\end{align}
\no where primes denote differentiation with respect to $\mu_1$. When $\varepsilon = 0$, \eqref{34} has positive solution
\begin{align}
    \mu_{1,*_0} = 2\sqrt{\frac{e_{BW}}{e_{2}}},
\end{align}
\no coinciding with the first-order correction of the most unstable Floquet exponent \eqref{24.75}.  When $0<\varepsilon\ll 1$, we expect $\mu_{1,*}$ to bifurcate smoothly from $\mu_{1,*_0}$. Since the small parameter in \eqref{34} appears as $\varepsilon^2$, we expand $\mu_{1,*}$ in $\varepsilon^2$, yielding 
\begin{align}
    \mu_{1,*} =  \mu_{1,*_0} + \varepsilon^2\mu_{1,*_2} + \mathcal{O}\left(\varepsilon^4\right). \label{35}
\end{align} 
\no Substituting \eqref{35} into \eqref{34}, 
\begin{align}
    \mu_{1,*_2} = -\frac{\lambda_{4,R}'(\alpha,\mu_{1,*_0})}{\lambda_{2,R}''(\alpha,\mu_{1,*_0})}. \label{35.25}
\end{align}
\no at $\mathcal{O}\left(\varepsilon^2\right)$. To simplify notation further, we drop the functional dependencies above, denoting  $\lambda_{4,R}'(\alpha,\mu_{1,*_0})$ and $\lambda_{2,R}''(\alpha,\mu_{1,*_0})$ instead by $\lambda_{4,R,*}'$ and $\lambda_{2,R,*}''$, respectively. Substituting \eqref{35.25} into \eqref{13}, we arrive at an asymptotic expansion for the Floquet exponent of the most unstable eigenvalue on the higher-order half-loop:
\begin{align}
    \mu_*
    &= \left(2\sqrt{\frac{e_{BW}}{e_{2}}}\right)\varepsilon - \left(2\frac{\Lambda_{4,R}^{(1,2,M)}}{\Lambda_{4,R}^{(1,1,M)}}\sqrt{\frac{e_{BW}}{e_{2}}}+\frac{\lambda_{4,R,*}'}{\lambda_{2,R,*}''}\right)\varepsilon^3+ \mathcal{O}\left(\varepsilon^4\right). \label{35.3}
\end{align}
\no If instead we substitute \eqref{35} into \eqref{33}, we obtain an asymptotic expansion for the real part of the most unstable eigenvalue. Using our simplified notation above,
\begin{equation}
    \lambda_{R,*}=\lambda_{2,R,*}\varepsilon^2 + \biggr( \lambda_{2,R,*}'\mu_{1,*_2}+\lambda_{4,R,*}\biggr)\varepsilon^4 +\mathcal{O}\left(\varepsilon^5\right).
\end{equation}
\no Unpacking this notation, we obtain the more explicit expansion
\begin{align}
    \lambda_{R,*} &= \frac{e_{BW}}{2}\varepsilon^2 + \frac{e_{BW}}{256c_0\Lambda_{4,R}^{(1,1,M)}e_2^2T_{2,0}^4}\biggr(c_0e_2\left(\Lambda_{4,R}^{(1,1,M)}\Lambda_{4,R,*}^{(1,2)}-\Lambda_{4,R}^{(1,2,M)}\Lambda_{4,R,*}^{(1,1)} \right)\phantom)\nonumber  \\
    &\phantom(\quad\quad\quad - 4\frac{T_{2,0}}{e_{BW}}\left(\Lambda_{4,R}^{(1,1,M)}\Lambda_{4,R,*}^{(2,2)}-\Lambda_{4,R}^{(1,2,M)}\Lambda_{4,R,*}^{(2,1)} \right)\biggr)\varepsilon^4 + \mathcal{O}\left(\varepsilon^5\right),  \label{35.4}
\end{align}
\no where $\Lambda_{4,R,*}^{(j,\ell)}$ denotes $\Lambda_{4,R}^{(j,\ell)}$ evaluated at $\mu_1 = \mu_{1,*_0}$.
A similar calculation determines the asymptotic expansion for the imaginary part of this eigenvalue. After some work, 
\begin{equation}
    \lambda_{I,*} = -2c_g\varepsilon\sqrt{\frac{e_{BW}}{e_{2}}} + \left(-c_g\left(\frac{\lambda_{4,R,*}'}{\lambda_{2,R,*}''} \right)+2\sqrt{\frac{e_{BW}}{e_{2}}}\left(\frac{c_g\Lambda^{(1,2,M)}_{4,R}}{\Lambda^{(1,1,M)}_{4,R}} + \frac{\Lambda_{3,I,*}}{32e_{2}T_{2,0}^2}\right) \right)\varepsilon^3+\mathcal{O}\left(\varepsilon^4\right),\label{35.45}
\end{equation}
\no where $\Lambda_{3,I,*}$ denotes $\Lambda_{3,I}$ evaluated at $\mu_1 = \mu_{1,*_0}$. Expansions \eqref{35.3},\eqref{35.4}, and \eqref{35.45} for the most unstable eigenvalue on the figure-eight match numerical computations to excellent agreement, even for sizeable values of $\varepsilon$ on the order of $0.2$. These expansions also improve upon results obtained at $\mathcal{O}\left(\varepsilon^2\right)$, see Figure \ref{fig10}.

\begin{figure}[tb]
    \centering \hspace*{-0.5cm}
    \includegraphics[height=7.2cm,width=15.5cm]{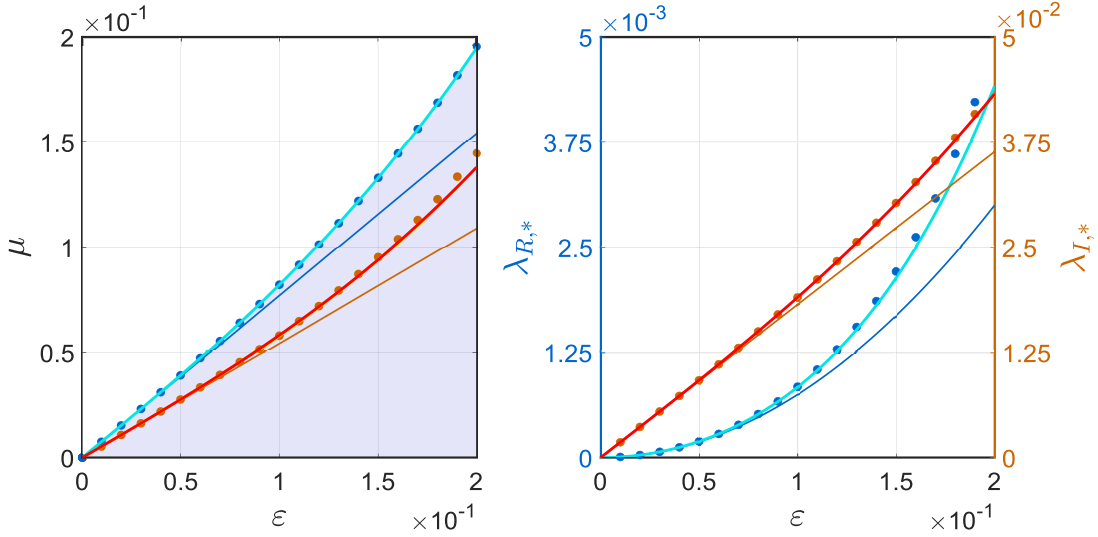} \vspace*{2mm}
    \caption{\small (Left) The interval of Floquet exponents parameterizing the half-loop of the Benjamin-Feir figure-eight curve for a Stokes wave with aspect ratio $\alpha = 1.5$ and variable amplitude $\varepsilon$. The numerically computed boundary of this interval is given by the blue dots, while the solid blue and light blue curves give the asymptotic results to $\mathcal{O}\left(\varepsilon\right)$ and $\mathcal{O}\left(\varepsilon^3\right)$, respectively. The orange dots give the numerically computed Floquet exponents of the most unstable eigenvalue, while the solid orange and red curves give the corresponding asymptotic estimates to $\mathcal{O}\left(\varepsilon\right)$ and $\mathcal{O}\left(\varepsilon^3\right)$, respectively. (Right) The real (blue axis) and imaginary (orange axis) part of the most unstable eigenvalue on the half-loop with $\alpha = 1.5$ and variable $\varepsilon$. Numerical results are given by the correspondingly colored dots. The asymptotic approximations of the real part to $\mathcal{O}\left(\varepsilon^2\right)$ and $\mathcal{O}\left(\varepsilon^4\right)$ are given by the solid blue and light blue curves, respectively. The asymptotic approximations of the imaginary part to $\mathcal{O}\left(\varepsilon\right)$ and $\mathcal{O}\left(\varepsilon^3\right)$ are given by the solid orange and red curves, respectively.}    \label{fig10}
\end{figure}

\subsection{Comparison of the Benjamin-Feir and High-Frequency Instabilities}

High-frequency instabilities of Stokes waves were first explored numerically by  Deconinck \emph{\&} Oliveras \cite{deconinckOliveras2011}. Unlike the Benjamin-Feir instability, these instabilities give rise to unstable spectra away from the origin in the complex spectral plane (Figure~\ref{fig2}), resulting in high-frequency spatial oscillations that are not commensurate or nearly commensurate with the fundamental period of the Stokes waves. A complete asymptotic analysis of high-frequency instabilities has been explored by the authors with Trichtchenko \cite{creedonDeconinckTrichtchenko2022}. Since then, \cite{hurYang2020} has validated some of these formal asymptotic results rigorously. 

For sufficiently small $\varepsilon$, the largest high-frequency instability is closest to the origin. As shown in \cite{creedonDeconinckTrichtchenko2022}, its most unstable eigenvalue has asymptotic expansion
\begin{align}
    \lambda_{R,*}^{(\textrm{HF})} = \frac{|\mathcal{S}_2|}{2\sqrt{\omega(k_0)\omega(k_0+2)}}\varepsilon^2 + \mathcal{O}\left(\varepsilon^4\right), \label{36}
\end{align}
\no where $\mathcal{S}_2$ is a complicated, but explicit, function of the aspect ratio $\alpha$, $\omega$ is given by \eqref{9}, and $k_0$ is an implicit function of $\alpha$ defined as the unique solution of
\begin{align}
    \Omega_{1}(k_0)=\Omega_{-1}(k_0+2), \label{36.5}
\end{align}
\no for $\Omega_{\sigma}$ in \eqref{9}. The leading-order behavior of this instability is $\mathcal{O}\left(\varepsilon^2\right)$, similar to the Benjamin-Feir case:
\begin{align}
    \lambda_{R,*}^{(\textrm{BFI})} = \frac{e_{BW}}{2}\varepsilon^2 + \mathcal{O}\left(\varepsilon^4\right). \label{37}
\end{align}
By comparing coefficients of the leading-order terms in \eqref{36} and \eqref{37}, we can directly compare the largest growth rates of the high-frequency and Benjamin-Feir instabilities for all $\alpha > 0$, see Figure \ref{fig11}. To our knowledge, this is the first time the growth rates of these two instabilities have been compared using analytical methods. A numerical comparison was is available in \cite{deconinckOliveras2011}.

For shallow water, $\alpha < \alpha_{BW}$, only high-frequency instabilities are present. For deep water, $\alpha > \alpha_{BW}$, we have two distinct behaviors. When $\alpha_{BW}<\alpha<\alpha_{DO}$, the high-frequency instabilities dominate the Benjamin-Feir instabilities. When $\alpha>\alpha_{DO}$, the Benjamin-Feir instability dominates. The critical threshold $\alpha_{DO}$ that distinguishes these behaviors in deep water is well-approximated by the implicit solution of
\begin{align}
    \frac{|\mathcal{S}_2|}{\sqrt{\omega(k_0)\omega(k_0+2)}} = e_{BW}, \label{38}
\end{align}
\no for $k_0$ defined in \eqref{36.5}. If we solve \eqref{38} numerically, we find $\alpha_{DO} = 1.4308061674...$, matching the numerical result presented in \cite{deconinckOliveras2011} to four significant digits. 
\begin{remark}
If the $\mathcal{O}\left(\varepsilon^4\right)$ corrections are included in expansions \eqref{36} and \eqref{37}, then $\alpha_{DO} = \alpha_{DO}(\varepsilon)$, where $\alpha_{DO}(0) = 1.4308061674...$. Using dominant balance, we argue that the next correction of $\alpha_{DO}$ appears at $\mathcal{O}\left(\varepsilon^2\right)$, but obtaining this correction explicitly is a computational challenge.  
\end{remark}

\begin{figure}[tb]
    \centering \hspace*{-0.5cm}
    \includegraphics[height=6.5cm,width=13.4cm]{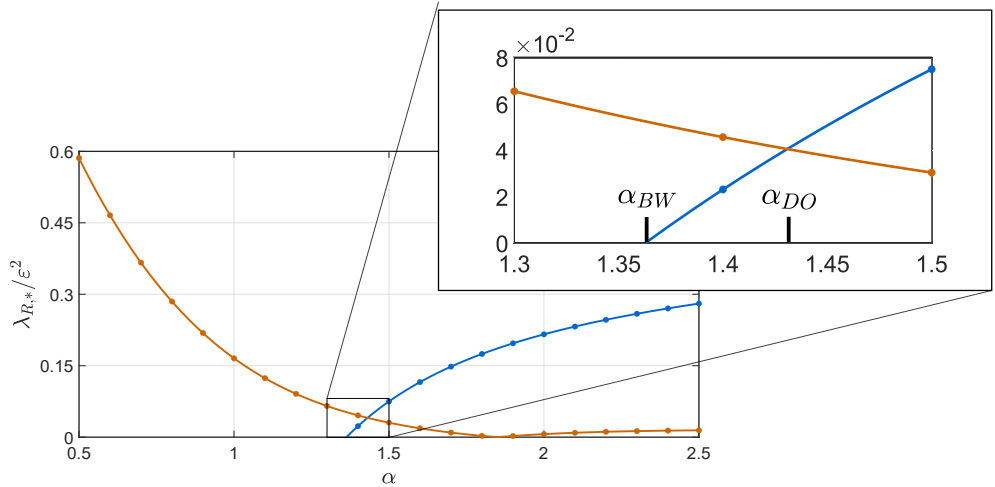} \vspace*{1mm}
    \caption{\small The real part of the most unstable eigenvalue (modulo $\varepsilon^2$) on the largest high-frequency instability (orange) and the Benjamin-Feir instability (blue) as a function of $\alpha$. Numerical results using $\varepsilon = 10^{-3}$ are given by the correspondingly colored dots. The asymptotic results \eqref{36} and \eqref{37} are given by the correspondingly colored solid curves. We observe three regimes for the periodic water wave problem: (i) $\alpha<\alpha_{BW}$, (ii) $\alpha_{BW}<\alpha<\alpha_{DO}$, and (iii) $\alpha>\alpha_{DO}$, where $\alpha_{BW}$ is the root of \eqref{23.1} and $\alpha_{DO}$ is the root of \eqref{38}. This agrees with the numerical results in Figure 11 of \cite{deconinckOliveras2011}.}    \label{fig11}
\end{figure}


\section{The Benjamin-Feir Spectrum in Infinite Depth}

\subsection{A Few Remarks about Infinite Depth}

Our analysis so far applies to the Benjamin-Feir instability spectrum in water of finite depth. In this section, we consider the special case of infinite depth, when $\alpha \rightarrow \infty$. Unfortunately, it is not possible to let $\alpha\rightarrow \infty$ using the expressions obtained in Section 3. Indeed, we have already seen that this limit is singular in the eigenfunction coefficients $\beta_{0,\pm1}$, which affects the description of the figure-eight curve at $\mathcal{O}\left(\varepsilon^3\right)$ and subsequent orders. The peculiarities of this limit are what cause the failure of the proof by Bridges \emph{\&} Mielke \cite{bridgesMielke1995} in infinitely deep water, necessitating the alternative proof by Nguyen \emph{\&} Strauss \cite{nguyenStrauss2020}. This also leads to the qualitative differences in the exact representations of the figure-eight curve in finite and infinite depth, as seen in \cite{bertiMasperoVentura2021,bertiMasperoVentura2022}. 

In this section, we outline the steps of our spectral perturbation method applied to the Benjamin-Feir instability in infinite depth, starting with the unperturbed problem. Replacing the finite-depth operators in the spectral problem \eqref{6} with their infinite-depth equivalents, see Subsection 2.3, we find the following spectral data at $\mathcal{O}\left(\varepsilon^0\right)$:
\begin{subequations}
\begin{align}
    \lambda_0 = 0, \quad  \quad \mu_0 &= 0, \quad \quad \textrm{and}\\
    {\bf w_0}(x) =  \beta_{0,-1}{\bf w_{0,-1}}(x)+&\beta_{0,0}{\bf w_{0,0}}(x)+ \beta_{0,1}{\bf w_{0,1}}(x) ,
\end{align}
\end{subequations}
\no where $\beta_{0,j}\in\mathbb{C}$ are undetermined at this order and
\begin{align}
    {\bf w_{0,-1}}(x) = \begin{pmatrix} 1 \\ i \end{pmatrix}e^{-ix}, \quad {\bf w_{0,0}}(x) = \begin{pmatrix} 0 \\ 1 \end{pmatrix}, \quad \textrm{and} \quad {\bf w_{0,1}}(x) = \begin{pmatrix} 1 \\ -i \end{pmatrix}e^{ix}.
\end{align}
\no A direct calculation shows that ${\bf w_0}$ above is in fact the limit of ${\bf w_0}$ in finite depth \eqref{11.5} as $\alpha \rightarrow \infty$. 

To proceed to higher order, we expand the infinite-depth operators as power series in $\varepsilon$. Some of these operators involve expressions of the form $|n+\mu|$, where $n \in \mathbb{Z}$. When we expand $\mu$ using \eqref{13}, it becomes necessary to expand $|n+\mu|$ as well. To obtain these expansions, we exploit the following identity:
\begin{align}
    |a+b| = |a| + \text{sgn}(a)b,  \label{39}
\end{align}
\no provided $a\in\mathbb{R}\setminus\{0\}$ and $b\in\mathbb{R}$ such that $|b|<|a|$. Substituting \eqref{13} into $|n+\mu|$, equating $\mu_0 = 0$, and applying \eqref{39} for sufficiently small $\varepsilon$ yields 
\begin{align}
    |n+\mu|  = \begin{cases} |n| + \text{sgn}(n)\varepsilon\mu_1(1+r(\varepsilon)), & n \neq 0, \\ |\varepsilon\mu_1|(1+r(\varepsilon)), & n=0. \end{cases} \label{40}
\end{align}
\no Consequently, all infinite-depth operators involving $|n+\mu|$ require two expansions: one for $n\neq0$ and one for $n=0$.

\subsection{The $\mathcal{O}\left(\varepsilon\right)$ Problem}
The $\mathcal{O}\left(\varepsilon\right)$ problem in infinite depth takes the same form as \eqref{14} with the finite-depth operators replaced by their infinite-depth equivalents and the expansions of these operators carried out appropriately. Three solvability conditions are obtained from this problem. One results in a trivial equality, similar to finite depth, and the remaining two are
\begin{subequations}
\begin{align}
    \beta_{0,-1}\left(\lambda_1 -i\frac{\mu_1}{2}\right) &= 0, \\
    \beta_{0,1}\left(\lambda_1 - i\frac{\mu_1}{2}\right) &= 0.
\end{align}
\end{subequations}
\no Imposing $\beta_{0,\pm}\neq1$ as in finite depth, we obtain 
\begin{align}
    \lambda_1 = i\frac{\mu_1}{2},
\end{align}
\no which is consistent with \eqref{18}, since $c_g \rightarrow -1/2$ as $\alpha \rightarrow \infty$. 

Having satisfied the solvability conditions, we solve the $\mathcal{O}\left(\varepsilon\right)$ problem for the first-order eigenfunction correction in infinite depth,
\begin{align}
    {\bf w_1}(x) = {\bf w_{1,p}}(x) +  \beta_{1,-1}{\bf w_{0,-1}}(x)+\beta_{1,0}{\bf w_{0,0}}(x)+ \beta_{1,1}{\bf w_{0,1}}(x) ,
\end{align}
\no where the coefficients $\beta_{1,j}\in \mathbb{C}$ are undetermined at this order and
\begin{equation}
\begin{aligned}
    {\bf w_{1,p}}(x) &=  \beta_{0,-1}\begin{pmatrix} 1 \\ i \end{pmatrix}e^{-2ix} + \frac12i\mu_1\beta_{0,-1}\begin{pmatrix} 0 \\ 1\end{pmatrix}e^{-ix} + \frac12i\mu_1\beta_{0,0}\begin{pmatrix} 1 \\ 0 \end{pmatrix} + \frac12 i\mu_1\beta_{0,1}\begin{pmatrix} 0 \\ 1 \end{pmatrix}e^{ix} \\ &\quad\quad\quad+ \beta_{0,1}\begin{pmatrix} 1 \\ -i \end{pmatrix}e^{2ix}, \label{41}
\end{aligned}
\end{equation}
\no which coincides with the $\alpha \rightarrow \infty$ limit of the corresponding particular solution in finite depth.
\subsection{The $\mathcal{O}\left(\varepsilon^2\right)$ Problem}

Similar to \eqref{19} in finite depth, this problem has three nontrivial solvability conditions 
\begin{subequations}
\begin{align}
    2\beta_{0,-1}\left(\lambda_2-i\frac{1}{2}r_1\mu_1\right) + i\left(-\beta_{0,1}+\left( \frac14\mu_1^2-1\right)\beta_{0,-1} \right) &= 0, \label{42a} \\
    \beta_{0,0}\mu_1^2 &= 0,  \label{42b}\\
    2\beta_{0,1}\left(\lambda_2-i\frac{1}{2}r_1\mu_1\right) + i\left(\beta_{0,-1}+\left( -\frac14\mu_1^2+1\right)\beta_{0,1} \right) &= 0. \label{42c}
\end{align}
\end{subequations}
Equations \eqref{42a} and \eqref{42c} are the limits of their respective equations \eqref{20a} and \eqref{20c} in the finite-depth case, since
\begin{align}
    \quad S_{2,-1} \rightarrow 0,~ T_{2,-1} \rightarrow \frac14,~ U_{2,-1} \rightarrow -1,~ V_{2,-1} \rightarrow -1,
\end{align}
as $\alpha \rightarrow \infty$. The same is true for the second equation \eqref{42b} if one divides the finite-depth equation \eqref{20b} by $T_{2,0}$ first. Then, because
\begin{align}
   S_{2,0} \rightarrow 1, ~   T_{2,0} \rightarrow -\infty,
\end{align}
as $\alpha \rightarrow \infty$, the rescaled \eqref{20b} tends to \eqref{42b}. The unbounded growth of $T_{2,0}$ as $\alpha \rightarrow \infty$ is the reason for the differences between the finite depth and infinite depth calculations, as mentioned in Subsection 3.3. 

Equation \eqref{42b} implies $\mu_1 = 0$ or $\beta_{0,0}=0$. Both numerical \cite{deconinckOliveras2011} and rigorous results \cite{bertiMasperoVentura2021,nguyenStrauss2020} suggest $\mu_1 \neq 0$, and we choose $\beta_{0,0} = 0$. 

\begin{remark}
We avoided normalizing ${\bf w}$ at the start of our analysis because $\beta_{0,0}=0$ in infinite depth. Indeed, had we chosen a non-zero normalization for $\beta_{0,0}$, we would need to renormalize our asymptotic expansions in infinite depth to avoid inconsistencies at higher order. 
\end{remark}

The remaining solvability conditions \eqref{42a} and \eqref{42c} form a nonlinear system of two equations in the three unknowns $\lambda_2$ and $\beta_{0,\pm 1}$. Without loss of generality, we choose $\beta_{0,-1}$ as a free parameter and solve for $\lambda_2$ and $\beta_{0,1}$. As in finite depth, we restrict our analysis to $\mu_1>0$. Solving for $\lambda_2$,
\begin{align}
    \lambda_2 = \lambda_{2,R} + i\lambda_{2,I},
\end{align}
\no with
\begin{subequations}
\begin{align}
   \lambda_{2,R} &= \pm \frac{\mu_1}{8}\sqrt{8-\mu_1^2}, \label{43a} \\
   \lambda_{2,I} &= \frac12r_1\mu_1. \label{43b}
\end{align}
\end{subequations}
\no Equations \eqref{43a} and \eqref{43b} are the limits of \eqref{22a} and \eqref{22b}, respectively, since
    \begin{align}
      e_{2} \rightarrow 1, ~  e_{BW} \rightarrow 1,
\end{align}
\no and $c_g \rightarrow -1/2$ as $\alpha \rightarrow \infty$. For $\lambda_2$ to have a nonzero real part, we must have
\begin{align}
    0<\mu_1<2\sqrt{2}, \label{44}
\end{align}
\no which is consistent with \eqref{24} as $\alpha \rightarrow \infty$. Floquet exponents satisfying \eqref{44} parameterize a single loop of the Benjamin-Feir figure-eight curve in the upper-half complex plane. If we repeat our analysis with $\mu_1<0$, we find $-2\sqrt{2}<\mu_1<0$, which parameterizes the remaining loop of the figure-eight. Combined, the full parameterizing interval of the figure-eight curve in infinite depth is
\begin{align}
    \mu \in \varepsilon\left(-2\sqrt{2},2\sqrt{2} \right)\left(1+r_1\varepsilon\right) + \mathcal{O}\left(\varepsilon^3\right). \label{45}
\end{align}
\no To simplify the remaining analysis, we restrict to a half-loop of the figure-eight curve by choosing the positive branch of \eqref{43a}, as in finite depth.

The imaginary correction \eqref{43b} and Floquet parameterization \eqref{45} depend on the first-order rescaling parameter $r_1$, similar to infinite depth. Using the regular curve condition at the next order, we find $r_1=-\sqrt{2}$. Assembling our expansions for the real and imaginary parts of the half loop in infinite depth, \begin{subequations}
\begin{align}
    \lambda_{R} &= \frac{\mu_1}{8}\varepsilon^2\sqrt{8-\mu_1^2} + \mathcal{O}\left(\varepsilon^3 \right), \label{45.1} \\
    \lambda_{I} &= \frac{1}{2}\mu_1\varepsilon -\frac{1}{\sqrt{2}}\mu_1\varepsilon^2 + \mathcal{O}\left(\varepsilon^3\right). \label{45.2}
\end{align}
\end{subequations}
These expansions agree well with numerical computations for sufficiently small $\varepsilon$, see Figure \ref{fig12}. Dropping $\mathcal{O}\left(\varepsilon^2\right)$ terms in these expansions and eliminating the $\mu_1$ dependence yields the curve
\begin{align}
    4\left(-1+4\sqrt{2}\varepsilon-12\varepsilon^2+8\sqrt{2}\varepsilon^3 - 4\varepsilon^4 \right)\lambda_{R}^2 = 2\varepsilon^2\left(-1+4\sqrt{2}\varepsilon-4\varepsilon^2 \right)\lambda_{I}^2 + \lambda_{I}^4, \label{46}
\end{align}
\no which is a lemniscate of Huygens, similar to finite depth. The coefficients of this lemniscate have additional dependence on $\varepsilon$ since $r_1 \neq 0$, in contrast with the finite-depth case. The low-order approximation of this curve obtained in \cite{bertiMasperoVentura2021} assumes $r_1 = 0$. This approximation works well enough for sufficiently small $\varepsilon$ but is not asymptotic to the true figure-eight curve to $\mathcal{O}\left(\varepsilon^2\right)$, see Figure \ref{fig13}. 

\begin{figure}[tb]
    \centering \hspace*{-0.0cm}
    \includegraphics[height=7.2cm,width=15cm]{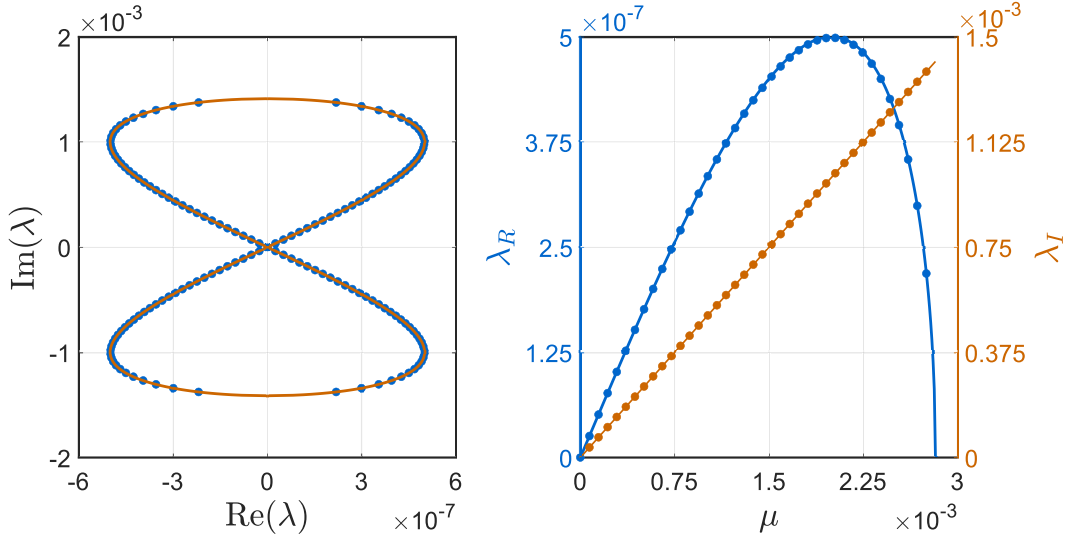} \vspace*{2mm}
    \caption{\small (Left) A plot of the Benjamin-Feir figure-eight curve for a Stokes wave with amplitude $\varepsilon = 10^{-3}$ in infinite depth. Numerical results are given by the blue dots, and the asymptotic results to $\mathcal{O}\left(\varepsilon^2\right)$ are given by the solid orange curve. (Right) The Floquet parameterization of the real (blue axis) and imaginary (orange axis) part of the figure-eight curve on the left. The respective numerical results are given by the correspondingly colored dots, and the asymptotic results for the real and imaginary part to $\mathcal{O}\left(\varepsilon^2\right)$ and $\mathcal{O}\left(\varepsilon\right)$, respectively, are given by the correspondingly colored curves.}    \label{fig12}
\end{figure}

\begin{remark}
The Floquet parameterization of the Benjamin-Feir instability in finite depth is
\begin{align} \mu \in \varepsilon\mu_1\left(1+r_2(\alpha)\varepsilon^2\right)+\mathcal{O}\left(\varepsilon^4\right),\end{align}
while in infinite depth,  
\begin{align} \mu \in \varepsilon\mu_1\left(1-\varepsilon\sqrt{2} \right) + \mathcal{O}\left(\varepsilon^3\right).\end{align}
\no In order for these parameterizations to be consistent, the corrective term $r_2(\alpha)\varepsilon^2$ in finite depth must be promoted an order of magnitude in $\varepsilon$ as $\alpha \rightarrow \infty$. Since $\varepsilon$ can be made arbitrarily small, the only way this is possible is if $|r_2|\rightarrow \infty$ as $\alpha \rightarrow \infty$, which is precisely what we observed in Subsection 3.3.
\end{remark}

\begin{figure}[tb]
    \centering \hspace*{-0.0cm}
    \includegraphics[height=5.9cm,width=12.5cm]{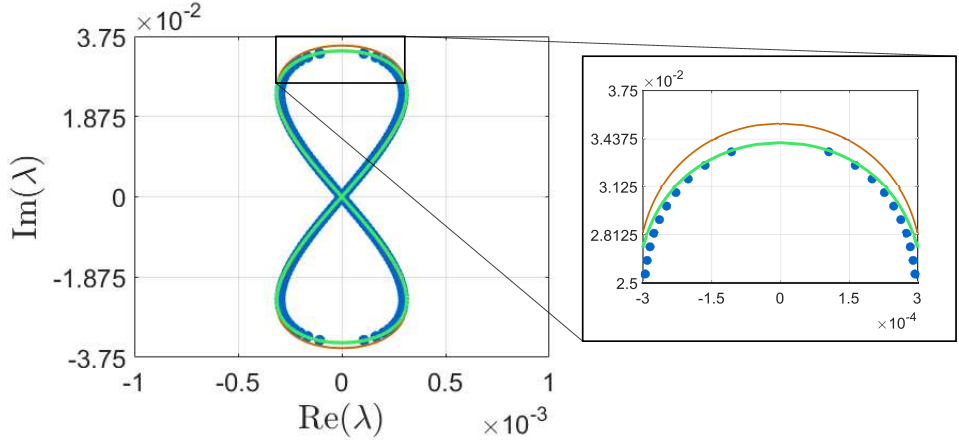} \vspace*{2mm}
    \caption{\small A plot of the Benjamin-Feir figure-eight curve for a Stokes wave with amplitude $\varepsilon = 2.5 \times 10^{-2}$ in infinite depth. Numerical results are given by the blue dots, and the asymptotic results to $\mathcal{O}\left(\varepsilon^2\right)$ are given by the solid green curve. The solid orange curve gives the approximation of Berti \emph{et al.} \cite{bertiMasperoVentura2021}. This approximation is not asymptotic in the imaginary part of the figure-eight curve to $\mathcal{O}\left(\varepsilon^2\right)$.}    \label{fig13}
\end{figure}

A direct calculation shows that \eqref{45.1} attains a maximum value of
\begin{align}
    \lambda_{R,*} = \frac12\varepsilon^2 + \mathcal{O}\left(\varepsilon^3\right), \label{45.4}
\end{align}
\no when $\mu_{1,*} =  2$. Hence, \eqref{45.4} gives an asymptotic expansion for the real part of the most unstable eigenvalue on the half-loop. Its corresponding imaginary part and Floquet exponent are
\begin{subequations}
\begin{align}
    \lambda_{I,*} &= \varepsilon + \mathcal{O}\left(\varepsilon^2\right), \\
    \mu_* &= 2\varepsilon + \mathcal{O}\left(\varepsilon^2\right), \label{45.5}
\end{align}
\end{subequations}
\no respectively. These expansions are consistent with those in finite depth (Subsection 3.2) as well as numerical results (Figure \ref{fig14}). 
 
\begin{figure}[tb]
    \centering \hspace*{-0.5cm}
    \includegraphics[height=7.2cm,width=15.5cm]{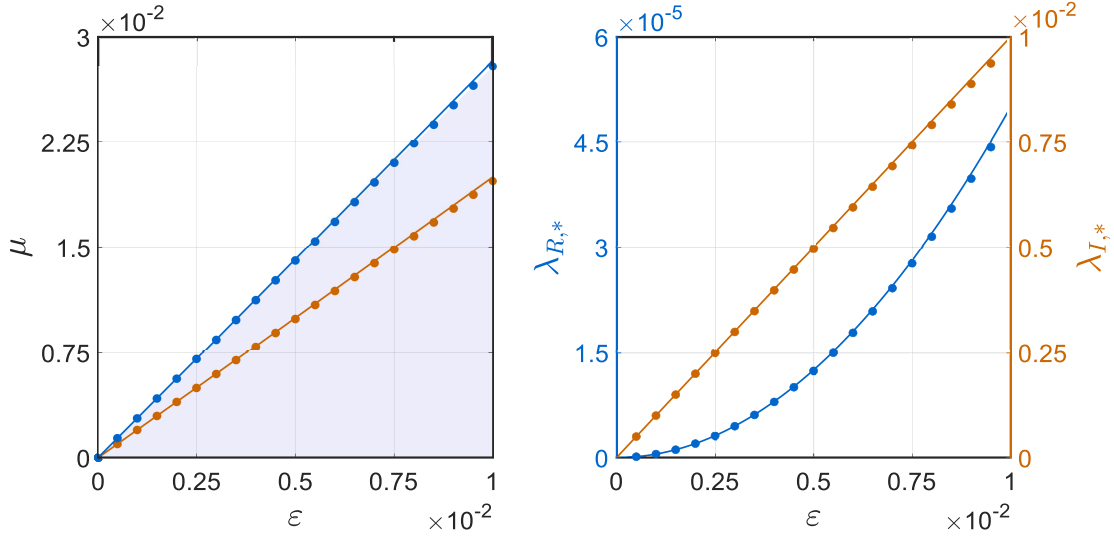} \vspace*{2mm}
    \caption{\small (Left) The interval of Floquet exponents parameterizing the half-loop of the Benjamin-Feir figure-eight curve for a Stokes wave with in infinite depth and variable amplitude $\varepsilon$. The numerically computed boundary of this interval is given by the blue dots, while the solid blue curve gives the asymptotic results to $\mathcal{O}\left(\varepsilon^2\right)$. The orange dots give the numerically computed Floquet exponents of the most unstable eigenvalue, while the solid orange curves give the corresponding asymptotic results to $\mathcal{O}\left(\varepsilon\right)$. (Right) The real (blue axis) and imaginary (orange axis) part of the most unstable eigenvalue in infinite depth with variable $\varepsilon$. Numerical results are given by the correspondingly colored dots, and the asymptotic results for the real and imaginary part to $\mathcal{O}\left(\varepsilon^2\right)$ and $\mathcal{O}\left(\varepsilon\right)$, respectively, are given by the correspondingly colored solid curves.}    \label{fig14}
\end{figure}

Continuing our analysis of the $\mathcal{O}\left(\varepsilon^2\right)$ problem, we solve \eqref{42a} and \eqref{42b} for $\beta_{0,1}$ (assuming $\mu_1>0$). After some work, we find
\begin{align}
    \beta_{0,1} = \frac14\Big(\mu_1^2-4\mp 4i\mu_1\sqrt{8-\mu_1^2} \Big)\beta_{0,-1}. \label{47}
\end{align}
Since we chose the positive branch of $\lambda_{2,R}$ \eqref{43a} without loss of generality, we choose the negative branch of \eqref{47}. In finite depth, $|\beta_{0,\pm 1}| \rightarrow \infty$ as $\alpha \rightarrow \infty$, which is inconsistent with \eqref{47}. However, a direct calculation shows that the ratio $\beta_{0,1}/\beta_{0,-1}$ in finite depth tends to \eqref{47} as $\alpha \rightarrow \infty$, re-establishing consistency between the two results.

Before we proceed to the next order, we solve the $\mathcal{O}\left(\varepsilon^2\right)$ problem subject to the solvability conditions \eqref{42a}-\eqref{42b}. We find 
\begin{align}
    {\bf w_2}(x) = {\bf w_{2,p}}(x) +  \beta_{2,-1}{\bf w_{0,-1}}(x)+\beta_{2,0}{\bf w_{0,0}}(x)+ \beta_{2,1}{\bf w_{0,1}}(x) ,
\end{align}
where $\beta_{2,j} \in \mathbb{C}$ are undetermined at this order and 
\begin{align}
    {\bf w_{2,p}}(x) = \sum_{j=-3}^{3}{\bf w_{2,j}}e^{ijx},
\end{align}
where ${\bf w_{2,j}} = {\bf w_{2,j}}(\beta_{0,-1},\beta_{1,\nu},r_1) \in \mathbb{C}^2$, see the companion Mathematica files for details. 

\subsection{The $\mathcal{O}\left(\varepsilon^3\right)$ Problem}

The solvability conditions of the $\mathcal{O}\left(\varepsilon^3\right)$ problem form a $3\times3$ linear system
\begin{align}
    \mathcal{M}\begin{pmatrix} \beta_{1,0} \\ \lambda_3 \\ \beta_{1,1} \end{pmatrix} = \begin{pmatrix} f_{3,1} \\ f_{3,2} \\ f_{3,3}\end{pmatrix}, \label{48}
\end{align}
\no where $\mathcal{M}$ is given by
\begin{align}
    \mathcal{M} = \begin{pmatrix} \mu_1 & 2\beta_{0,-1} & -i \\ -\mu_1^2 & 0 & 0 \\ -\mu_1 & 2\beta_{0,1} & -\frac{i}{4}\left(-4+8i\lambda_2+4r_1\mu_1+\mu_1^2 \right)\end{pmatrix}, \label{48.1}
\end{align}
\no and
\begingroup
\allowdisplaybreaks
\begin{subequations}
\begin{align}
    f_{3,1} &= \frac{i}{4}\beta_{1,-1}\biggr(4+8i\lambda_2+4r_1\mu_1-\mu_1^2 \biggr), \\
    f_{3,2} &= -i\mu_1^2\biggr(\beta_{0,-1}+\beta_{0,1} \biggr), \\
    f_{3,3} &= -i\biggr(\beta_{1,-1} + \frac{1}{4}\beta_{0,1}\mu_1\Big(6-4r_2+4i\lambda_2+\mu_1^2\Big) \biggr).
\end{align}
\end{subequations}
\endgroup
\no If we substitute expressions for $\lambda_2$ and $\beta_{0,1}$ on the half-loop into \eqref{48.1}, we find 
\begin{align}
    \textrm{det}\left(\mathcal{M}\right) = \mu_1^3\beta_{0,-1}\sqrt{8-\mu_1^2},
\end{align}
\no implying \eqref{48} has a unique set of solutions for $0<\mu_1<2\sqrt{2}$, as desired. The solution for $\lambda_3$ is
\begin{align}
    \lambda_{3} = \lambda_{3,R}+i\lambda_{3,I},
\end{align}
\no with
\begin{subequations}
\begin{align}
    \lambda_{3,R} &= -\frac{\mu_1\left(2\mu_1+r_1\left(-4+\mu_1^2 \right) \right)}{4\sqrt{8-\mu_1^2}}, \label{49a}\\
    \lambda_{3,I} &= -\frac{1}{16}\mu_1\left(16-8r_2+\mu_1^2\right). \label{49b}
\end{align}
\end{subequations}
\no To avoid singular behavior in $\lambda_{3,R}$ as $\mu_1 \rightarrow 2\sqrt{2}$, we choose $r_1$ such that
\begin{align}
    \lim_{\mu_1 \rightarrow 2\sqrt{2}}\Big(2\mu_1+r_1(-4+\mu_1^2) \Big) = 0,
\end{align}
\no according to the regular curve condition. We find $r_1 = -\sqrt{2}$, justifying our prior claim.

Given $r_1=-\sqrt{2}$, we see that \eqref{49a} and \eqref{49b} are nonzero for generic choices of $0<\mu_1<2\sqrt{2}$. Hence, we have obtained a higher-order correction to both the real and imaginary parts of the figure-eight curve in infinite depth. This contrasts with the finite-depth case, where an imaginary correction only was found at $\mathcal{O}\left(\varepsilon^3\right)$. 

To characterize this higher-order correction, it is necessary to determine the value of the second-order rescaling parameter $r_2$ appearing in \eqref{49b}. We show at the next order that $r_2 = 13/8$, using the regular curve condition. Assuming this for now, we assemble our expansions for the real and imaginary parts along a half-loop of this higher-order curve 
\begin{subequations}
\begin{align}
    \lambda_R &= \frac{1}{8}\mu_1\varepsilon^2\sqrt{8-\mu_1^2}\left(1 + \left(\frac{2(-2\mu_1+\sqrt{2}(-4+\mu_1^2))}{8-\mu_1^2} \right)\varepsilon \right) + \mathcal{O}\left(\varepsilon^4\right), \label{50a}\\
    \lambda_I &= \frac12\mu_1\varepsilon\Big(1 -\sqrt{2}\varepsilon - \frac{1}{8}\big(3+\mu_1^2\big)\varepsilon^2\Big) +\mathcal{O}\left(\varepsilon^4 \right), \label{50b}
\end{align}
\end{subequations}
\no respectively. The interval of Floquet exponents for the entire curve has asymptotic expansion
\begin{align}
    \mu \in \varepsilon\left(-2\sqrt{2},2\sqrt{2}\right)\left(1-\sqrt{2}\varepsilon+\frac{13}{8}\varepsilon^2 \right) + \mathcal{O}\left(\varepsilon^4\right). \label{51}
\end{align}
\no These expansions agree well with numerical computations for sufficiently small $\varepsilon$, see Figure \ref{fig15}. In theory, one could eliminate the dependence of $\mu_1$ from \eqref{50a} and \eqref{50b} to obtain an algebraic curve that approximates the true figure-eight to $\mathcal{O}\left(\varepsilon^3\right)$, but this process is cumbersome and provides little insight into the behavior of the true figure-eight curve.

\begin{figure}[tb]
    \centering \hspace*{-0.0cm}
    \includegraphics[height=7.2cm,width=15cm]{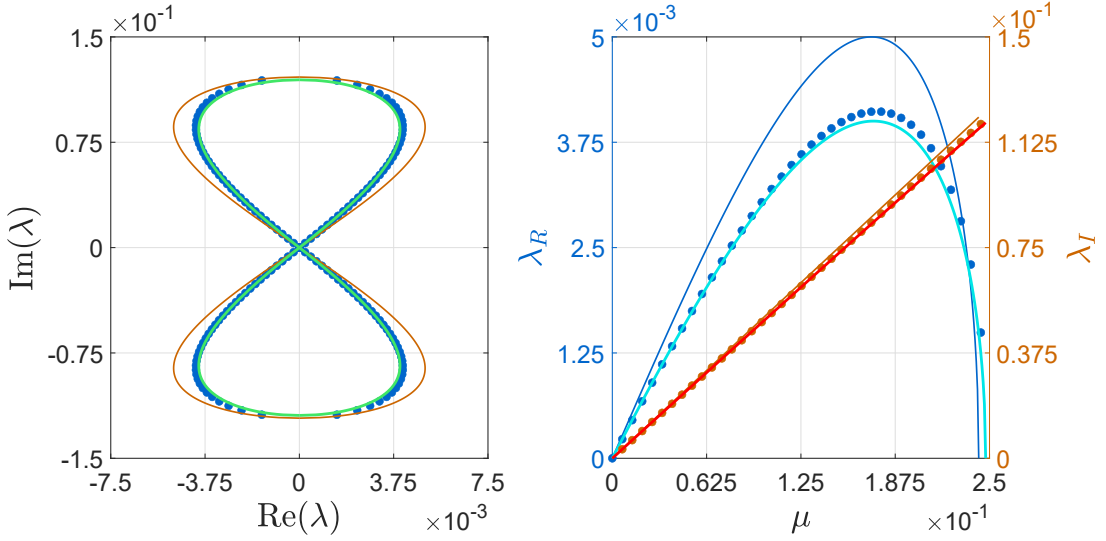} \vspace*{2mm}
    \caption{\small (Left) A plot of the Benjamin-Feir figure-eight curve for a Stokes wave with amplitude $\varepsilon = 0.1$ in infinite depth. Numerical results are given by the blue dots, while asymptotic results to $\mathcal{O}\left(\varepsilon^2\right)$ and $\mathcal{O}\left(\varepsilon^3\right)$ are given by the solid orange and green curves, respectively. (Right) The Floquet parameterization of the real (blue axis) and imaginary (orange axis) part of the figure-eight curve on the left. Numerical results are given by the correspondingly colored dots. The asymptotic parameterizations of the real part to $\mathcal{O}\left(\varepsilon^2\right)$ and $\mathcal{O}
    \left(\varepsilon^3\right)$ are given by the solid blue and light blue curves, respectively, while those for the imaginary part are given by the solid orange and red curves, respectively. }    \label{fig15}
\end{figure}

Proceeding as in Section 3.4, we can derive an asymptotic expansion for the most unstable eigenvalue on this figure-eight curve and for its corresponding Floquet exponent. In particular, if we let $\mu_{1,*}$ denote a critical point of \eqref{50a}, then
\begin{align}
   \frac{\partial}{\partial \mu_1} \Biggr(\frac{1}{8}\mu_1\varepsilon^2\sqrt{8-\mu_1^2}\left(1 + \left(\frac{2(-2\mu_1+\sqrt{2}(-4+\mu_1^2))}{8-\mu_1^2} \right)\varepsilon \right) + \mathcal{O}\left(\varepsilon^4\right)\Biggr)\biggr|_{\mu_1 = \mu_{1,*}} = 0.
\end{align}
\no Dropping terms of $\mathcal{O}\left(\varepsilon^4\right)$ and smaller, we arrive at the following equation for $\mu_{1,*}$: 
\begin{align}
   32-12\mu_{1,*}^2+\mu_{1,*}^4 - \varepsilon\big(32\sqrt{2}+32\mu_{1,*}-24\mu_{1,*}^2\sqrt{2}-2\mu_{1,*}^3+2\mu_{1,*}^4\sqrt{2}  \big) = 0. \label{52}
\end{align}
\no When $\varepsilon = 0$, the only positive solution of \eqref{52} is $\mu_{1,*_0} = 2$, which coincides with \eqref{45.5} from the previous order. When $0<\varepsilon\ll 1$, 
\begin{align}
    \mu_{1,*} = \mu_{1,*_0} + \varepsilon\mu_{1,*_1} + \mathcal{O}\left(\varepsilon^2\right), \label{53}
\end{align}
\no since $\varepsilon$ appears as the small parameter in \eqref{52}. Substituting \eqref{53} into \eqref{52}, we find at $\mathcal{O}\left(\varepsilon\right)$ that $\mu_{1,*_1} = -3+2\sqrt{2}.$ Thus, the Floquet exponent of the most unstable eigenvalue on the figure-eight has asymptotic expansion
\begin{align}
    \mu_* &=  \left(2 + (-3+2\sqrt{2})\varepsilon + \mathcal{O}\left(\varepsilon^2\right)\right)\varepsilon\left(1 -\varepsilon\sqrt{2}+\frac{13}{8}\varepsilon^2 +\mathcal{O}\left(\varepsilon^3\right)\right), 
\end{align}
\no which simplifies to
\begin{align*}
    \mu_* = 2\varepsilon -3\varepsilon^2 + \mathcal{O}\left(\varepsilon^3\right).
\end{align*}
\no Substituting \eqref{53} into \eqref{50a} and \eqref{50b}, we obtain asymptotic expansions 
\begin{subequations}
\begin{align}
    \lambda_{R,*} &= \frac{1}{2}\varepsilon^2 - \varepsilon^3 + \mathcal{O}\left(\varepsilon^4\right), \\
    \lambda_{I,*} &= \varepsilon -\frac32 \varepsilon^2 + \mathcal{O}\left(\varepsilon^3\right),
\end{align}
\end{subequations}
\no for the real and imaginary part of this most unstable eigenvalue on the half-loop, respectively. These expansions agree well with numerical computations (Figure \ref{fig16}), although not to the same degree as the corresponding results in finite depth. This is a result of resolving the higher-order figure-eight curve in infinite depth at $\mathcal{O}\left(\varepsilon^3\right)$ as opposed to $\mathcal{O}\left(\varepsilon^4\right)$. To complete our analysis of the solvability conditions \eqref{48}, we report solutions for $\beta_{1,0}$ and $\beta_{1,1}$ on the half-loop with $r_1 = -\sqrt{2}$. We find
\begin{subequations}
\begin{align}
\beta_{1,0} &= \mu_1\beta_{0,-1}\left(\sqrt{8-\mu_1^2}+\frac{i}{4}\mu_1 \right), \\
    \beta_{1,1} &= \frac18\left( \mu_1\left(-4-4\mu_1\sqrt{2}+3\mu_1^2\right)\beta_{0,-1} + 2\left(-4+\mu_1^2\right)\beta_{1,-1}\right) \nonumber \\
    &\quad\quad\quad+\frac{i\mu_1}{8\sqrt{8-\mu_1^2}}\left(\left(16\sqrt{2}+\mu_1\left(-16-4\mu_1\sqrt{2}+3\mu_1^2\right)\right)\beta_{0,-1}+2\left(-8+\mu_1^2\right)\beta_{1,-1}\right), 
\end{align}
\end{subequations}
\no where $\beta_{0,-1}$ and $\beta_{1,-1}$ depend on the normalization of ${\bf w}$.

\begin{figure}[tb]
    \centering \hspace*{-0.5cm}
    \includegraphics[height=7.2cm,width=15.5cm]{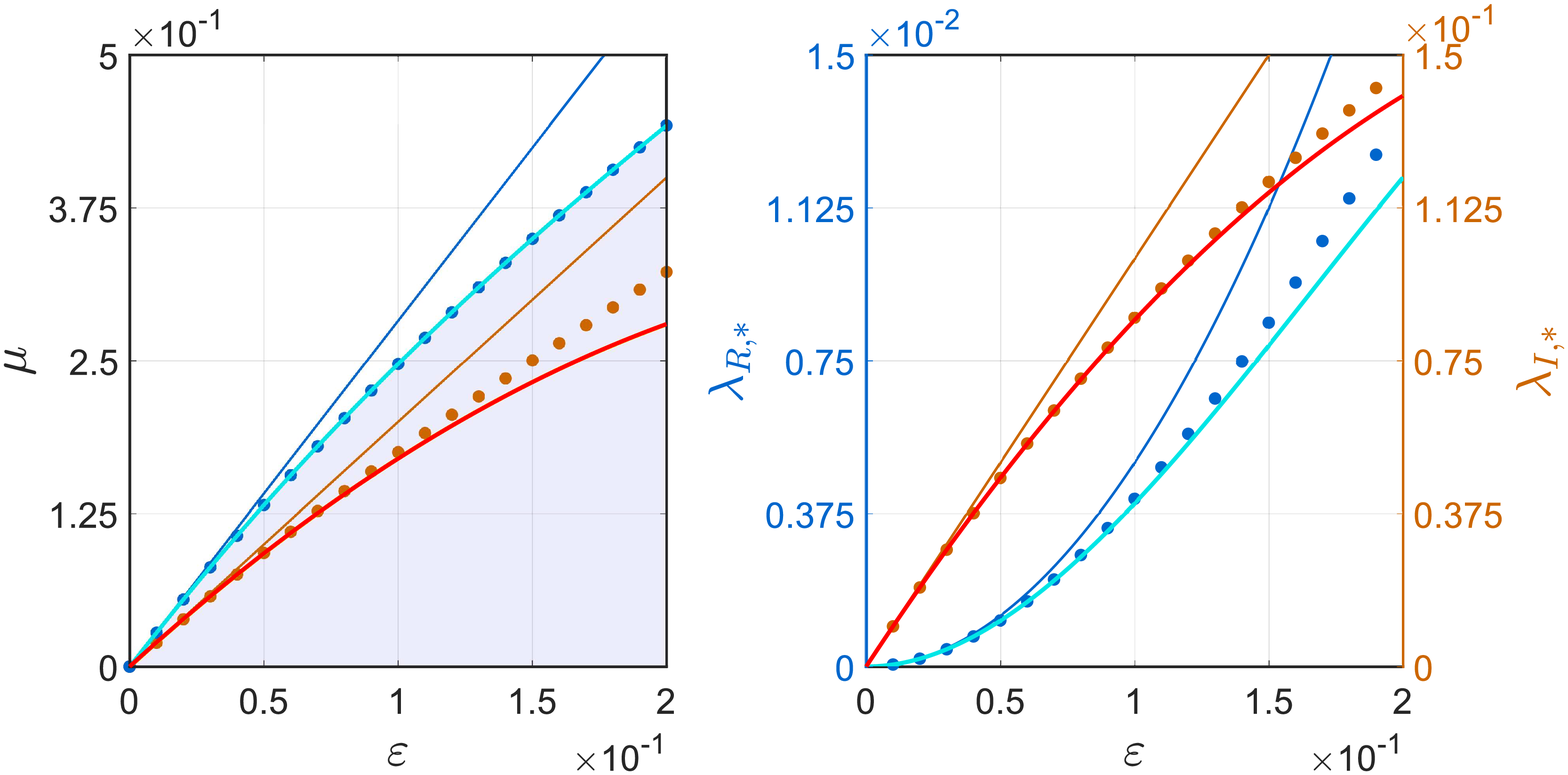} \vspace*{2mm}
    \caption{\small (Left) The interval of Floquet exponents parameterizing the half-loop of the Benjamin-Feir figure-eight curve for a Stokes wave in infinite depth with variable amplitude $\varepsilon$. The numerically computed boundary of this interval is given by the blue dots, while the solid blue and light blue curves give the asymptotic results to $\mathcal{O}\left(\varepsilon\right)$ and $\mathcal{O}\left(\varepsilon^3\right)$, respectively. The orange dots give the numerically computed Floquet exponents of the most unstable eigenvalue, while the solid orange and red curves give the corresponding asymptotic estimates to $\mathcal{O}\left(\varepsilon\right)$ and $\mathcal{O}\left(\varepsilon^2\right)$, respectively. (Right) The real (blue axis) and imaginary (orange axis) part of the most unstable eigenvalue on the half-loop in infinite depth with variable $\varepsilon$. Numerical results are given by the correspondingly colored dots. The asymptotic approximations of the real part to $\mathcal{O}\left(\varepsilon^2\right)$ and $\mathcal{O}\left(\varepsilon^3\right)$ are given by the solid blue and light blue curves, respectively. The asymptotic approximations of the imaginary part to $\mathcal{O}\left(\varepsilon\right)$ and $\mathcal{O}\left(\varepsilon^2\right)$ are given by the solid orange and red curves, respectively.}    \label{fig16}
\end{figure}

Finally, we solve the $\mathcal{O}\left(\varepsilon^3\right)$ problem subject to the solvability conditions \eqref{48}, arriving at an expression for the third-order eigenfunction correction
\begin{align}
    {\bf w_3}(x) = \sum_{j=-4}^{4}{\bf w_{3,j}}e^{ijx} +  \beta_{3,-1}{\bf w_{0,-1}}(x)+\beta_{3,0}{\bf w_{0,0}}(x)+ \beta_{3,1}{\bf w_{0,1}}(x)
\end{align}
where ${\bf w_{3,j}} = {\bf w_{3,j}}(\beta_{0,-1},\beta_{1,-1},\beta_{2,\nu},r_2) \in \mathbb{C}^2$  while $\beta_{3,j} \in \mathbb{C}$ are undetermined constants at this order, see the companion Mathematica files for more details. 
\subsection{The $\mathcal{O}\left(\varepsilon^4\right)$ Problem}
The solvability conditions at $\mathcal{O}\left(\varepsilon^4\right)$ also form a $3\times3$ linear system
\begin{align}
    \mathcal{M}\begin{pmatrix} \beta_{2,0} \\ \lambda_4 \\ \beta_{2,1} \end{pmatrix} = \begin{pmatrix} f_{4,1} \\ f_{4,2} \\ f_{4,3} \end{pmatrix}, \label{54}
\end{align}
\no where $\mathcal{M}$ is the same as before and
\begingroup
\allowdisplaybreaks
\begin{subequations}
 \begin{align}
 f_{4,1} &= -\biggr(\beta_{2,-1}\biggr(2\lambda_2+\frac{i}{4}\left(-4+\mu_14\sqrt{2}+\mu_1^2\right) \biggr) +\beta_{1,-1}\biggr(2\lambda_3 + \frac{\mu_1}{4}\big(-4ir_2 + 4\lambda_2  \phantom) \phantom)\phantom) \nonumber \\ &\phantom(\phantom(\quad\quad\quad\hspace*{-0.28cm}+  i\big(6+\mu_1^2\big)\big)\biggr) - \frac14\mu_1\beta_{1,0}\left(4\sqrt{2}+\mu_1\right)+\beta_{0,-1} \biggr(-2i-i\lambda_2^2 + \lambda_3\mu_1+\lambda_2\mu_1^2 \phantom)\nonumber \\
    &\phantom(\phantom(\quad\quad\quad\hspace*{-0.28cm}-\frac{i}{8}\mu_1\left(12\sqrt{2}+8r_3+\mu_1\left(3+2\mu_1\left(\sqrt{2}-\mu_1 \right) \right) \right) \biggr)+ i\beta_{0,1}\left(-2+\frac{3}{16}\mu_1^2 \right)\biggr),  \\
    f_{4,2} &= -i\mu_1^2\biggr(\beta_{1,-1}-\frac{i}{4}\beta_{1,0}\big(8\sqrt{2}+\mu_1 \big)+\beta_{1,1}-\frac12\beta_{0,-1}\big(4\sqrt{2}+\mu_1\big) - 2\beta_{0,1}\sqrt{2} \biggr),  \\
    f_{4,3} &= -\biggr(i\beta_{2,-1} + \mu_1\beta_{1,0}\big(\sqrt{2}-\frac{1}{4}\mu_1 \big)+\beta_{1,1}\biggr(2\lambda_3 + \frac{i}{4}\mu_1\left(6-4r_2+4i\lambda_2+\mu_1^2 \right) \biggr) \phantom) \nonumber \\
    &\quad\quad\quad+i\beta_{0,-1}\big(2-\frac{3}{16}\mu_1^2 \big) + \beta_{0,1}\biggr(2i+i\lambda_2^2 + \lambda_2\mu_1^2-\frac{i}{8}\mu_1\big(12\sqrt{2} + 8r_3 - 8i\lambda_3\phantom) \phantom) \nonumber \\
    &\phantom(\phantom(\phantom(\phantom(\quad\quad\quad \hspace*{-0.54cm}+\mu_1\big(-3 \phantom) +2\mu_1\big(\sqrt{2}+\mu_1\big)\big)\big)\biggr)\biggr).
    \end{align}
    \end{subequations}
\endgroup
\no Solving for $\lambda_4$ on the half loop, we find
\begin{align}
    \lambda_4 = \lambda_{4,R}+i\lambda_{4,I},
\end{align}
\no with
\begin{subequations}
\begin{align}
    \lambda_{4,R} &= \mu_1\left(\frac{-2176+32r_2\left(32-12\mu_1^2+\mu_1^4\right)+\mu_1\left(1024\sqrt{2}+\mu_1\left(432-64\mu_1\sqrt{2} -92\mu_1^2+5\mu_1^4 \right) \label{55a} \right)}{128\left(8-\mu_1^2\right)^{3/2}}\right),\\
    \lambda_{4,I} &= \frac{1}{16}\mu_1\left(16\sqrt{2}+8r_3+\mu_1\left(8+3\mu_1\sqrt{2} \right) \right). \label{55b}
\end{align}
\end{subequations}
\no For ease of notation, let $\Lambda_{4,R}$ denote the numerator of \eqref{55a}. A direct calculation shows that $\Lambda_{4,R}$ factors as follows:
\begin{align}
   \Lambda_{4,R} &= \mu_1\left(2\sqrt{2}-\mu_1\right)\Big(-544\sqrt{2}+240\mu_1+168\mu_1^2\sqrt{2}+52\mu_1^3-10\mu_1^4\sqrt{2}-5\mu_1^5\phantom) \nonumber \\ &\phantom(\quad\quad\quad\hspace*{-0.45cm}-32r_2\left(2\sqrt{2}+\mu_1\right)\left(-4+\mu_1^2\right) \Big). \label{56}
\end{align}
\no It appears that \eqref{55a} already satisfies the regular curve condition, since $\Lambda_{4,R} \rightarrow 0$ as $\mu_1 \rightarrow 2\sqrt{2}$. However, the factor of $8-\mu_1^2$ in the denominator of \eqref{55a} is one power larger than at the previous order \eqref{49a}. Thus, we cannot guarantee regular behavior of $\lambda_{4,R}$ if only the first factor of \eqref{56} tends to zero as $\mu_1 \rightarrow 2\sqrt{2}$. We must also impose similar behavior on the second factor:
\begin{align}
   \lim_{\mu_1\rightarrow2\sqrt{2}} \Big(-544\sqrt{2}+240\mu_1+168\mu_1^2\sqrt{2}+52\mu_1^3-10\mu_1^4\sqrt{2}-5\mu_1^5\phantom)-32r_2\left(2\sqrt{2}+\mu_1\right)\left(-4+\mu_1^2\right) \Big)=0. \label{57}
\end{align}
~\\
Solving \eqref{57} for $r_2$ gives us the desired result $r_2 = 13/8$. As a consequence, the final expression for the fourth-order real part correction \eqref{55a} becomes
\begin{align}
    \lambda_{4,R} &= \mu_1\left(\frac{-512+\mu_1\left(1024\sqrt{2}+\mu_1\left(-192+\mu_1\left(-64\sqrt{2}+5\mu_1\left(-8+\mu_1^2 \right) \right)\right) \right)}{128\left(8-\mu_1^2\right)^{3/2}}\right). \label{58}
\end{align}
Since \eqref{55b} and \eqref{58} are generically nonzero for $0<\mu_1<2\sqrt{2}$, we have found another higher-order approximation to a half loop of the figure-eight curve in infinite depth, up to the unknown third-order rescaling parameter $r_3$. Presumably, one can determine this value at $\mathcal{O}\left(\varepsilon^5\right)$ using the techniques presented in this section. We stop here, since we have already obtained a higher-order approximation to the figure-eight curve at the previous order.
\begin{remark}
For the sake of completeness, the final expressions of $\beta_{2,0}$ and $\beta_{2,1}$ solving \eqref{54} are found in the companion Mathematica files.
\end{remark}

\section{Conclusions}

Building on work by Akers \cite{akers2015} and collaborations with Trichtchenko \emph{et al.} \cite{creedonDeconinckTrichtchenko2021b,creedonDeconinckTrichtchenko2021a,creedonDeconinckTrichtchenko2022}, we have developed a formal perturbation method to compute high-order asymptotic approximations of the Benjamin-Feir figure-eight curve, present in the stability spectrum of small-amplitude Stokes waves in water of sufficient depth. Unlike traditional methods in spectral perturbation theory \cite{kato1966}, this method allows us to approximate the entire curve at once. 

Using our method, we are able to determine 
\begin{enumerate}
    \item[(i)] the Floquet exponents that parameterize the figure-eight curve,
    \item[(ii)] the real and imaginary parts of the most unstable eigenvalue on the figure-eight curve, and
    \item[(iii)] algebraic curves asymptotic to the figure-eight curve.
\end{enumerate}
We compare these expressions directly with numerical computations of the figure-eight curve using methods presented in \cite{deconinckOliveras2011}. To our knowledge, this is the first time a numerical and analytical description of the Benjamin-Feir instability have been compared. Excellent agreement between these descriptions is found in finite and infinite depth, even for modest values of the Stokes wave amplitude $\varepsilon$. Our expressions are also consistent with the rigorous results of Berti \emph{et al.} \cite{bertiMasperoVentura2021,bertiMasperoVentura2022} and their heuristic approximations of the figure-eight curve.

In addition, our asymptotic results elucidate key differences between the behavior of the Benjamin-Feir instability spectrum in finite and infinite depth. In particular, the first-order rescaling parameter $r_1$ for the Floquet parameterization of the figure-eight curve vanishes in finite depth, while  $r_1 =- \sqrt{2}$ in infinite depth. Consequently, the second-order rescaling correction $r_2$ is singular in finite depth as one approaches infinitely deep water, \emph{i.e.}, as $\alpha \rightarrow \infty$. This singularity propagates to the imaginary part of the figure-eight curve at third order and to the real part of the figure-eight curve at fourth order. Thus, the limit as $\alpha \rightarrow \infty$ is singular for the Benjamin-Feir instability, illustrating the breakdown of compactness mentioned in \cite{bertiMasperoVentura2022,bridgesMielke1995,nguyenStrauss2020}. 

Using asymptotic results in this work and in \cite{creedonDeconinckTrichtchenko2022}, we are able to compare the Benjamin-Feir instability and the most unstable high-frequency instability for the first time analytically. Our analysis suggests three natural regimes for the water wave problem:
\begin{itemize}
    \item[(i)] Shallow water, which occurs when $\kappa h < \alpha_{BW} = 1.3627827567...$ and only high-frequency instabilities are present, 
    \item[(ii)] Intermediate water, which occurs when $\alpha_{BW}<\kappa h <\alpha_{DO}(\varepsilon) = 1.4308061674... + \mathcal{O}\left(\varepsilon^2\right)$ and both instabilities are present, but high-frequency instabilities dominate,
    \item[(iii)] Deep water, which occurs when $\kappa h > \alpha_{DO}(\varepsilon)$ and both instabilities are present, but the Benjamin-Feir instability dominates.
\end{itemize}
Here, $\kappa$ is the wavenumber of the Stokes wave and $h$ is the depth of the water. These regimes are supported by numerical computations in \cite{deconinckOliveras2011}. We conclude that Stokes waves of all depths and all wavenumbers are unstable to the Benjamin-Feir instability, high-frequency instabilities, or both.


\section*{Data Availability Statement}

The asymptotic expressions derived in this work can be found in the companion Mathematica files \emph{wwp\textunderscore bf\textunderscore  fd.nb} (for finite depth) and \emph{wwp\textunderscore bf\textunderscore id.nb} (for infinite depth). 

\section*{Acknowledgements} 

We acknowledge useful conversations with Massimiliano Berti, Huy Nguyen, and Walter Strauss. R.C. was supported in part by funding from an ARCS Foundation Fellowship and from the Ruth Jung Chinn Fellowship in Applied Mathematics at the University of Washington. 

\section*{Declaration of Interests}
The authors report no conflict of interest.

\bibliographystyle{plain}
\bibliography{bfpaper}

\end{document}